\documentclass[reqno,12pt]{amsart}
\newtheorem{theorem}{Theorem}[section]
\newtheorem{theorem*}[theorem]{Theorem}
\newtheorem{lemma}[theorem]{Lemma}

\newtheorem{proposition}[theorem]{Proposition}

\theoremstyle{remark}
\newtheorem{remark}[theorem]{Remark}
\numberwithin{equation}{section}
\theoremstyle{theorem}

\usepackage{array}
\usepackage{multirow}
\usepackage{array,booktabs}
\usepackage[margin=1in]{geometry}
\usepackage{verbatim}
\usepackage{amsmath}
\usepackage{amsfonts}
\usepackage{amssymb}
\usepackage{amsthm}
\usepackage{enumerate}
\usepackage{amsrefs}
\usepackage{amscd}
\usepackage{enumerate}
\usepackage{amsrefs}
\usepackage{epic}
\usepackage{mathtools}

\usepackage{amsaddr}

\usepackage{graphicx}
\usepackage[font=small,skip=0pt]{caption}

\begin{document}

\author{Georgios Antoniou, Misha Feigin}

\address{ School of Mathematics and Statistics, University of Glasgow, \\ University Place, 
Glasgow, G12 8QQ, UK}
\email{g.antoniou.1@research.gla.ac.uk, misha.feigin@glasgow.ac.uk }

\title{\bf{Supersymmetric V-systems}}

\begin{abstract}
We construct ${\mathcal N}=4$  $D(2,1;\alpha)$ superconformal quantum mechanical system for any configuration of vectors forming a $\vee$-system. In the case of a Coxeter root system the bosonic potential of the supersymmetric Hamiltonian is the corresponding generalised Calogero--Moser potential. We also construct supersymmetric generalised trigonometric Calogero--Moser--Sutherland Hamiltonians for some root systems including $BC_N$. 
\end{abstract}
\maketitle

\section{Introduction}
 
Calogero--Moser Hamiltonian is a famous example of an integrable system \cite{Calogero, Moser, Suther} which is related to a number of mathematical areas (see e.g. \cite{EP}). Generalised Calogero--Moser systems associated with an arbitrary root system were introduced by Olshanetsky and Perelomov \cite{OP1}, \cite{OP2}. $\mathcal{N}=2$ supersymmetric quantum Calogero--Moser systems were constructed in \cite{Freed} and considered  further  in \cite{Brink}. They were generalised to classical root systems in \cite{BTW} and to an arbitrary root system in \cite{MS}.

A motivation for construction of $\mathcal N = 4$ Calogero--Moser system goes back to the work \cite{TS} on a  conjectural description of near--horizon limit of Reissner-Nordstr\"om black hole where appearance of $su(1,1|2)$ superconformal Calogero--Moser model was suggested.  Though we also note more recent different considerations of near extremal black holes in \cite{MSY}. Another motivation to study supersymmetric (trigonometric) Calogero--Moser--Sutherland systems comes from the relation of these systems with conformal blocks and possible generalisation of these relations to the supersymmetric case \cite{ISh}.

 Wyllard gave an ansatz for $\mathcal N =4$ supercharges in \cite{Wyllard}. In general Wyllard's ansatz depends on two potentials $F$ and $W$. He constructed $su(1,1|2)$ $N$ particle Calogero--Moser Hamiltonian for a single value of the coupling parameter $c=1/N$ as bosonic part of his supersymmetric Hamiltonian with $W=0$. Wyllard argued  that his ansatz does not produce superconformal Calogero--Moser Hamiltonians for general values of $c$. 
Necessary differential equations for $F$ and $W$ were derived in \cite{Wyllard}. Thus potential $F$ satisfies generalised Witten--Dijkgraaf--Verlinde--Verlinde (WDVV) equations (in the form of \cite{MMM}) as it was pointed out in \cite{BGL}. Wyllard's potential $F$ has the form 
\begin{equation}
\label{potF}
F=\sum_{\gamma \in \mathcal A}  (\gamma, x)^2 \log (\gamma, x),
\end{equation}
where $\mathcal A$ is the root system $A_{N-1}$. Examples based on root systems $\mathcal A = G_2, B_3$ were also considered in \cite{Wyllard}.  
Solutions $F$ to WDVV equations of this type appear also in  Seiberg--Witten theory \cite{MMM} and in theory of Frobenius manifolds \cite{DubrovinAlmost}. 

More generally, Veselov introduced the notion of a $\vee$-system in \cite{Veselov}. $\vee$-systems form special collections of vectors in a linear space, which satisfy certain linear algebraic conditions. A logarithmic prepotential \eqref{potF} corresponding to a collection of vectors $\mathcal A$ satisfies WDVV equations if $\mathcal A$ is a $\vee$-system. The class of $\vee$--systems contains Coxeter root systems, deformations of generalized root systems of Lie superalgebras, special subsystems in and restrictions of such systems \cite{FV07, SV}. A complete description of the class remains open (see \cite{FV14} and references therein).

Several attempts have been made to construct supersymmetric mechanics such that the corresponding Hamiltonian has bosonic potential of Calogero--Moser type with a reasonably general coupling parameter(s).  
Wyllard's ansatz for $\mathcal N=4$ supercharges was extended to other root systems in \cite{GLP07}, \cite{GLP09} where solutions for a small number of particles were studied both for $W=0$ and $W\ne 0$. In particular, $su(1,1|2)$ superconformal Calogero--Moser systems related to $\mathcal A = A_1\oplus G_2, F_4$ and subsystems of $F_4$ were derived. Superconformal $su(1,1|2)$ Calogero--Moser systems for the rank two root systems were derived in \cite{BKSU} via suitable action in the superspace. For the WDVV equations arising in the superfield approach we refer to \cite{KLPO}.

A many-body model with ${D(2, 1 ; \alpha)}$ supersymmetry algebra with $\alpha=-\frac12$ was considered in \cite{fedorukMPOSP}. This model was obtained by a reduction from matrix model and it incorporates an extra set of bosonic variables (``${U}(2)$ spin variables'') which enter the bosonic potential of the corresponding Hamiltonian.
One-dimensional version of such a model was considered in \cite{FIL09May} and, for any $\alpha$, in \cite{BK09May}, \cite{fedoruk}. 
 A generalisation of the many-body classical spin superconformal model for any value of the parameter $\alpha$ was proposed in \cite{kriv}.  Within ${D(2, 1 ; \alpha)}$ supersymmetry ansatz of \cite{kriv} a class of bosonic potentials was obtained in \cite{FS}. The potential $F$ has the form \eqref{potF} for a root system $\mathcal A$. Then $W$ is a twisted period of the Frobenius manifold on the space of orbits corresponding to the root system $\mathcal A$. Such polynomial twisted periods were described in \cite{FS}, they exist for special values of parameter $\alpha$. Although the corresponding bosonic potentials are algebraic this class does not seem to contain generalised Calogero--Moser potentials associated with $\mathcal A$.

Recently  a construction of type $A_{N-1}$ supersymmetric (classical) Calogero--Moser model with extra spin bosonic generators and $\mathcal{N} N^2$ fermionic variables (for any even $\mathcal{N}$) was presented in \cite{kriv2}. The ansatz for supercharges is more involved and extra fermionic variables appear due to reduction from a matrix model.
A related quantum  $\mathcal N = 4$ supersymmetric spin $A_{N-1}$ Calogero--Moser system was studied recently in \cite{FILS}. Furthermore, a simpler ansatz for supercharges for the spin classical $A_{N-1}$ Calogero--Moser system was presented in \cite{KLS18}. This model has $\frac12 {\mathcal N} N(N+1)$ fermionic variables and the supersymmetry algebra is  $osp({\mathcal N}| 2)$.
Most recently classical supersymmetric $osp({\mathcal N}|2)$ Calogero--Moser systems 
 were presented in \cite{KLPS}; these models have nonlinear Hermitian conjugation property of matrix fermions and  supercharges are cubic in fermions.




In the current work we present two constructions of supersymmetric $\mathcal N=4$ quantum mechanical system starting with an arbitrary $\vee$-system. In the case of a Coxeter root system $\mathcal A$ the bosonic part of the Hamiltonian is the Calogero--Moser Hamiltonian associated with $\mathcal A$ introduced by Olshanetsky and Perelomov in \cite{OP2}, which we get in two different gauges: the potential and potential free ones. In the latter case the Hamiltonian is not formally self-adjoint; this gauge comes from  the radial part of the Laplace--Beltrami operator on symmetric spaces \cite{BPF, OP2, Helg}. The superconformal algebra is $D(2,1; \alpha)$ where $\alpha$ depends on the $\vee$-system and is ultimately related with the coupling parameter in the resulting Calogero--Moser type Hamiltonian. We use original ansatz for the supercharges \cite{Wyllard}, \cite{GLP07} based on the potentials $F$, $W$ and we take $W=0$. In the special case when $\alpha=-1$  the superalgebra $D(2,1; -1)$
contains the superalgebra $su(1,1|2)$ as its subalebra, and our first ansatz on the $su(1,1|2)$ generators reduces to the one considered in \cite{GLP07, GLP09}. It was emphasised in \cite{GLP09} that such quantum models with $W=0$ are non-trivial with bosonic potentials proportional to squared Planck constant, though they were not considered in more detail in \cite{GLP09}. Thus we extend considerations in \cite{GLP09} for $W=0$ to the case of superconformal algebra $D(2, 1 ; \alpha)$, and we get in this framework quantum Calogero--Moser type systems associated with an arbitrary $\vee$--system, which includes Olshanetsky--Perelomov generalisations of the Calogero--Moser system with arbitrary invariant coupling parameters.

We also consider generalised trigonometric Calogero--Moser--Sutherland systems related to a collection of vectors $\mathcal{A}$ with multiplicities. We include these Hamiltonians in the supersymmetry algebra provided that extra assumptions on $\mathcal{A}$ are satisfied which are similar to WDVV equations for the trigonometric version of the potential $F$. We show that these assumptions can be satisfied when $\mathcal{A}$ is an irreducible root system with more than one orbit of the Weyl group, that is $BC_N$, $F_4$ and $G_2$ cases. A related solution of WDVV equations for the root system $B_N$ was obtained in \cite{HM}.

The structure of the paper is as follows. We recall the definition of the Lie superalgebra $D(2,1; \alpha)$ in Section \ref{sal}. We give two types of representations of this superalgebra in Sections \ref{1st rep section}, \ref{2nd rep section}. Starting with any $\vee$-system we get two corresponding supersymmetric Hamiltonians. In Section \ref{Hamsect} we present them explicitly. We consider supersymmetric trigonometric Calogero--Moser--Sutherland systems in Section \ref{trigsect}.

\vspace{3mm}

{\bf Acknowledgments}. We are grateful to M. Alkadhem, S. Dubovsky, S. Fedoruk, A. Galajinsky, E. Ivanov, S. Krivonos,  O. Lechtenfeld, I. Strachan and  A. Sutulin for many useful discussions and comments. In particular, we would like to thank S. Krivonos for questions which led us to considerations in Section \ref{trigsect}. The work of Georgios Antoniou was funded by EPSRC doctoral training partnership grants
EP/M506539/1, EP/M508056/1, EP/N509668/1.

\section{The ${D(2, 1;\alpha)}$ Lie superalgebra}
\label{sal}

Let us recall the definition of the family of Lie superalgebras ${D(2, 1;\alpha)}$, which depends on a parameter $\alpha \in \mathbb C$ (see e.g.  \cite[Section 20]{sorba}). 
The algebra has $8$ odd generators $Q^{abc}$  and $9$ even generators $T^{ab}=T^{ba}, I^{ab}=I^{ba}, J^{ab}=J^{ba}$ ($a,b,c=1,2$). Elements  $T^{ab}$, $I^{ab}$ and $J^{ab}$ generate three pairwise commuting $sl(2)$ algebras. 

Let $\epsilon_{ab}$, $\epsilon^{ab}$ be the fully anti-symmetric tensors in two dimensions such that $ \epsilon_{12}=\epsilon^{21}=1$.
Then all the relations of the superalgebra $D(2,1; \alpha)$ take the following form:

%
%
%
\begin{equation}\label{odd}
\{ Q^{a c e}, Q^{b d f}\} = - 2 \big( \epsilon^{ef} \epsilon^{c d} T^{ab} + \alpha \epsilon^{ab} \epsilon^{c d} J^{ef} - (\alpha+1) \epsilon^{ab} \epsilon^{ef} I^{c d}\big),
\end{equation}
\begin{equation}\label{T}
[T^{ab}, T^{cd}]= -i \big( \epsilon^{ac} T^{bd} + \epsilon^{bd} T^{ac} \big),
\end{equation}
\begin{equation}\label{JI}
\textbf{a}) \, [J^{ab}, J^{cd}]= -i \big( \epsilon^{ac} J^{bd} + \epsilon^{bd} J^{ac} \big), \quad \textbf{b}) \, [I^{ab}, I^{cd}]= -i \big( \epsilon^{ac} I^{bd} + \epsilon^{bd} I^{ac} \big),
\end{equation}
\begin{equation}\label{odd-even}
\textbf{a}) \,[T^{ab}, Q^{cdf}]= i \epsilon^{c (a}Q^{b)df} , \quad \textbf{b}) \,[J^{ab}, Q^{cdf}]= i \epsilon^{f (a}Q^{|cd|b)}, \quad\textbf{c}) \, [I^{ab}, Q^{cdf}]= i \epsilon^{d (a}Q^{|c|b)f},
\end{equation}
where we symmetrise over two indices inside $(\dots)$ with indices inside $|\dots|$ being unchanged. For example, $\epsilon^{f (a}Q^{|cd|b)}= \frac{1}{2} \big( \epsilon^{fa} Q^{cdb} + \epsilon^{fb}Q^{cda}\big)$. 

We also have relations 
\begin{align}\label{thezerorelations}
[T^{ab}, I^{cd}]=[I^{cd}, J^{ef}]= [T^{ab}, J^{ef}]=0, 
\end{align}
for all $a,b, c, d, e, f=1,2$. Let us rename generators as follows:
\begin{equation*}
Q^{a}= - Q^{21a}, \quad \bar{Q}^a= -  Q^{22a}, \quad S^a=Q^{11a}, \quad\bar{S}^a= Q^{12a}, \quad a=1, 2,
\end{equation*}
\begin{equation*}
 K=T^{11} , \quad H=T^{22} , \quad D=-T^{12}=-T^{21}.
\end{equation*}

 We will use $\epsilon_{ab}$ and $\epsilon^{ab}$ to lower and raise indices, e.g. $Q^a =\epsilon^{ab}Q_b$, $\bar{Q}^a= \epsilon^{ab} \bar{Q}_b$. 

We consider $N$ (quantum) particles on a line with coordinates and momenta $(x_j, p_j)$,  $j=1, \dots, N$ to each of which we associate four fermionic variables $\{ \psi^{aj}, \bar{\psi}_a^j | a=1, 2\}$. We will also write $x=(x_1,\ldots, x_N)$, $p=(p_1,\ldots, p_N)$.

We assume the following (anti)-commutation relations ($a, b = 1, 2; j,k = 1, \dots, N$):
\begin{align}\label{extebracket1}
[x_j, p_k] = i\delta_{jk}, \quad \{ \psi^{aj}, \bar{\psi}_b^k \} =-\frac{1}{2} \delta^{jk} \delta^{a}_b, \quad \{\psi^{aj}, \psi^{bk}\}=\{  \bar{\psi}_a^j,  \bar{\psi}_b^k\}=0.
\end{align}
Thus one can think of $p_k$ as $p_k =- i \frac{\partial}{\partial x_k}$. 

We introduce further fermionic variables by 
\begin{align}\label{e^ab}
\psi_a^j= \epsilon_{ab} \psi^{bj}, \quad \bar{\psi}^{aj} = \epsilon^{ab} \bar{\psi}_b^j. 
\end{align}
They satisfy the following useful relations:
\begin{align}\label{extebracket2}
\{\psi_{a}^j,  \bar{\psi}^{bk}\}= \frac{1}{2} \delta^{jk} \delta^{b}_a, \quad \{ \psi^{aj},  \bar{\psi}^{bk}\}=  \frac{1}{2} \epsilon^{jk} \epsilon^{ab}, \quad \{ \psi_{a}^j,  \bar{\psi}_b^k\} =  \frac{1}{2} \delta^{jk} \epsilon_{ba}.
\end{align}

We will be assuming throughout that summation over repeated indices takes place (even when both indices are either low or upper indices) unless it is indicated that no summation is applied.

Let $F=F(x_1, \dots, x_N)$ be a function such that 
\begin{align}
\label{alphalambda}
x_r F_{rjk}= - (2 \alpha +1) \delta_{jk},
\end{align}
where $F_{rjk} = \frac{\partial^3 F}{\partial {x_r} \partial {x_j} \partial {x_k}}$ for any $r, j, k  = 1, \dots, N$. We assume that all the derivatives $F_{rjk}$ are homogeneous in $x$ of degree -1.  Furthermore, we assume that $F$ satisfies the following Witten-Dijkgraaf-Verlinde-Verlinde equations (WDVV) equations 
\begin{align}\label{WDVV}
F_{rjk}F_{kmn}= F_{rmk}F_{kjn}, 
\end{align}
for any $r, j, k, m,n = 1, \dots, N$.

The following relations for arbitrary operators $A$, $B$, $C$ will be useful:
\begin{align}
[AB, C] &= A[B, C] + [A, C] B, \label{com1} \\
[AB, C]&= A \{ B, C \} - \{A, C \} B, \label{com2} \\
\{AB, C\}&= A[B,C] + \{A, C \}B. \label{com3}
\end{align}

We are going to present two representations of $D(2,1; \alpha)$ algebra using $F$.

\section{The first representation} 
\label{1st rep section}

Let the supercharges be of the form
\begin{align}\label{Q}
Q^a& = p_r \psi^{ar} + i F_{rjk} \langle \psi^{br} \psi_b^j  \bar{\psi}^{ak} \rangle,
\end{align}
\begin{align}\label{Qbar}
\bar{Q}_c &= p_l  \bar{\psi}_c^l + i F_{lmn} \langle  \bar{\psi}_d^l  \bar{\psi}^{dm} \psi_c^n \rangle,
\end{align}
where  the symbol $\langle \dots \rangle$ stands for the anti-symmetrisation. That is given $N$  operators $A_i$, ($i=1, \dots, N$) we define
\begin{align}\label{weyl}
\langle A_1 \dots A_N \rangle = \frac{1}{N!} \sum_{\sigma \in S_N} \operatorname{sgn}(\sigma) A_{
\sigma(1)} \dots A_{\sigma(N)}. 
\end{align}
Note that we have by  (\ref{extebracket1}), (\ref{extebracket2}) and (\ref{weyl})
\begin{align*}
\langle \psi^{br} \psi_b^j  \bar{\psi}^{ak} \rangle &= \frac{1}{6} ( 2 \psi^{br} \psi_b^j  \bar{\psi}^{ak} + 2  \bar{\psi}^{ak} \psi^{br} \psi_b^j - \psi^{br}  \bar{\psi}^{ak} \psi_b^j + \psi_b^j  \bar{\psi}^{ak} \psi^{br}) \\
&= \frac{1}{3}( \psi^{br} \psi_b^j  \bar{\psi}^{ak} +  \bar{\psi}^{ak} \psi^{br} \psi_b^j - \psi^{br}  \bar{\psi}^{ak} \psi_b^j ) + \frac{1}{12} (\delta^{jk} \psi^{ar} - \delta^{rk} \psi^{aj})\\
&= \psi^{br}\psi_b^j  \bar{\psi}^{ak} - \frac{1}{6} \delta^{rk} \psi^{aj} - \frac{1}{3} \delta^{jk} \psi^{ar}+\frac{1}{12} (\delta^{jk} \psi^{ar} - \delta^{rk} \psi^{aj}). 
\end{align*}
Note that $F_{rjk} (\delta^{jk} \psi^{ar} - \delta^{rk} \psi^{aj})=0$ since $\delta^{jk} \psi^{ar} - \delta^{rk} \psi^{aj}$ is anti-symmetric under the interchange of $k$ and $r$. Note also that $F_{rjk} \psi^{aj} \delta^{rk}= F_{rjk} \psi^{ar} \delta^{jk}$. Therefore
\begin{align}\label{weylmanyparticle1rep}
 F_{rjk}\langle \psi^{br} \psi_b^j  \bar{\psi}^{ak} \rangle=  F_{rjk} ( \psi^{br} \psi_b^j  \bar{\psi}^{ak} - \frac{1}{2} \psi^{ar} \delta^{jk}).
\end{align}
Similarly,
\begin{align}\label{weylmanyparticle1rep2ndsuper}
F_{lmn} \langle  \bar{\psi}_d^l  \bar{\psi}^{dm} \psi_c^n \rangle = F_{lmn}(  \bar{\psi}_d^l  \bar{\psi}^{dm} \psi_c^n  - \frac{1}{2} \bar{\psi}_c^l \delta^{nm} ).
\end{align}
Let also
\begin{align}
\label{K}
K=x^2= \sum_{j=1}^N x_j^2, 
\end{align}
\begin{align}
\label{D}
 D= -\frac{1}{4} \{x_j, p_j\}=-\frac{1}{2} x_j p_j + \frac{i N}{2},
\end{align}
\begin{align}
\label{I}
I^{11}= - i \psi_a^j \psi^{aj}, \quad I^{22}= i \bar{\psi}^{aj} \bar{\psi}_a^j, \quad I^{12}=-\frac{i}{2} [ \psi_a^j, \bar{\psi}^{aj}],
\end{align}
\begin{align}
\label{J}
J^{ab}=J^{ba} = 2i \psi^{(aj}\bar{\psi}^{bj)},
\end{align}
\begin{align}
\label{S}
S^a= -2 x_j \psi^{aj}, \quad \bar{S}_a = -2 x_j \bar{\psi}_a^j.
\end{align}

\begin{remark}
Ansatz \eqref{Q}, \eqref{Qbar}, \eqref{K}, \eqref{D}, \eqref{J}, \eqref{S} with $F$ satisfying \eqref{alphalambda} at $\alpha =-1$ matches considerations in \cite{GLP07} (see also \cite{Wyllard}, \cite{GLP09}), where $su(1,1|2)$ superconformal mechanics was considered. Note that superalgebra $su(1,1|2)$ generated by $Q^{abc}, T^{ab}, J^{ab}$ is a subalgebra in the superalgebra $D(2, 1; -1)$. Thus Lemmas \ref{Jmanyparticle1}, \ref{oddevenbmanyparticle1}  below can be deduced from considerations in \cite{GLP07}. We include these lemmas so that to have complete derivations for reader's convenience.
\end{remark}

Let us firstly check relations (\ref{JI}), (\ref{odd-even}) involving generators $J^{ab}$ and $I^{ab}$.

\begin{lemma}[cf. \cite{GLP07}] \label{Jmanyparticle1} Let $J^{ab}$ be given by \eqref{J}. Then relations (\ref{JI}a) hold. 
\end{lemma}

\begin{proof}
We consider the commutator
\begin{align*}
[\psi^{aj} \bar{\psi}^{bj}, \psi^{ck} \bar{\psi}^{dk}]&= \psi^{aj} [ \bar{\psi}^{bj}, \psi^{ck} \bar{\psi}^{dk}] + [ \psi^{aj}, \psi^{ck} \bar{\psi}^{dk}] \bar{\psi}^{bj} \\
&= \frac{1}{2} \epsilon^{cb} \psi^{aj} \bar{\psi}^{dj} + \frac{1}{2} \epsilon^{da} \psi^{cj} \bar{\psi}^{bj},
\end{align*}
which implies the statement.
\end{proof}

We will use the following relations: 

\begin{align}\label{forproofofIS}
[ \bar{\psi}^{bk}, \psi_a^j \psi^{aj}]=\psi^{bk}, \quad [ \bar{\psi}^{aj}  \bar{\psi}_a^j, \psi_b^k]=- \bar{\psi}_b^k.
\end{align}
\begin{lemma}\label{Imanyparticle1}
Let $I^{ab}$ be given by \eqref{I}. Then relations (\ref{JI}b) hold.
\end{lemma}

\begin{proof}
The relations (\ref{JI}b) read 
\begin{align*}
[I^{11}, I^{22}]= 2i I^{12}, \quad [I^{11}, I^{12}]= i I^{11}, \quad [I^{22}, I^{12}]= -i I^{22}. 
\end{align*}
We have 
\begin{align}\label{I^11I22many1rep}
[I^{11}, I^{22}]&= [\psi_a^j \psi^{aj}, \bar{\psi}^{bk} \bar{\psi}_b^k].
\end{align}
By applying (\ref{com1}), (\ref{com2}) we rearrange expression (\ref{I^11I22many1rep}) as
\begin{align*}
[I^{11}, I^{22}]&= \psi_a^j [ \psi^{aj}, \bar{\psi}^{bk} \bar{\psi}_{b}^k] + [ \psi_a^j , \bar{\psi}^{bk} \bar{\psi}_b^k] \psi^{aj} \\
&= \psi_a^j \bar{\psi}^{aj} + \bar{\psi}_a^j \psi^{aj} = \psi_a^j \bar{\psi}^{aj} - \bar{\psi}^{aj}\psi_a^j \\
&= 2i I^{12},
\end{align*}
as required. Moreover, using the Jacobi identity we have
$$
[I^{11}, I^{12}] =-\frac{1}{2} [\psi_a^j \psi^{aj}, [\psi_b^k,  \bar{\psi}^{bk}]]= \frac{1}{2} [ \psi_b^k, [ \bar{\psi}^{bk}, \psi_a^j \psi^{aj}]].
$$
 Thus by using the first relation in (\ref{forproofofIS})  
\begin{align*}
[I^{11}, I^{12}] = \psi_b^k \psi^{bk}= i I^{11}.
\end{align*}
Similarly, $$ [I^{22}, I^{12}]=\frac{1}{2} [  \bar{\psi}^{aj}  \bar{\psi}_a^j ,  [\psi_b^k,  \bar{\psi}^{bk}]]=-\frac{1}{2} [\bar{\psi}^{bk}, [ \bar{\psi}^{aj}  \bar{\psi}_a^j, \psi_b^k]].$$ Hence, by using the latter relation in (\ref{forproofofIS})
\begin{align*}
[I^{22}, I^{12}]= \bar{\psi}^{bk} \bar{\psi}_b^k=-i I^{22},
\end{align*}
and hence the statement follows. 
\end{proof}

In what follows, we will use the following relation:
\begin{align}\label{QJmany1rep1form}
[\psi^{aj} \bar{\psi}^{bj}, \psi^{cl}]= - \frac{1}{2} \epsilon^{bc}  \psi^{al}.
\end{align}
By formulae (\ref{com1}), (\ref{com2}) we also have 
\begin{align}
[\psi^{aj} \bar{\psi}^{bj}, \psi^{dl} \psi_d^m \bar{\psi}^{cn}] &= \psi^{dl} \psi_d^m [ \psi^{aj} \bar{\psi}^{bj} , \bar{\psi}^{cn}] + [\psi^{aj} \bar{\psi}^{bj}, \psi^{dl} \psi_d^m ] \bar{\psi}^{cn}  \nonumber \\
&= - \psi^{dl} \psi_d^m  \bar{\psi}^{bj} \{  \bar{\psi}^{cn}, \psi^{aj}\} + \psi^{dl} [\psi^{aj}  \bar{\psi}^{bj}, \psi_d^m] \bar{\psi}^{cn} + [ \psi^{aj}  \bar{\psi}^{bj}, \psi^{dl}] \psi_d^m  \bar{\psi}^{cn} \nonumber\\
&= \frac{1}{2} \epsilon^{ca}  \psi^{dl} \psi_d^m  \bar{\psi}^{bn} + \frac{1}{2} \psi^{bl} \psi^{am}  \bar{\psi}^{cn}  +\frac{1}{2}   \psi^{bm}  \psi^{al} \bar{\psi}^{cn} \label{commutatorforI12} .
\end{align}

\begin{lemma}[cf. \cite{GLP07}]\label{oddevenbmanyparticle1}
Let $Q^{abc}$, $J^{ab}$ be as above. Then the relations (\ref{odd-even}b) hold.
\end{lemma}

\begin{proof}
Firstly let us note that the sum of the last two terms in (\ref{commutatorforI12}) is anti-symmetric in $a$ and $b$ and $J^{ab}=J^{ba}$. Therefore we have by applying (\ref{commutatorforI12})
\begin{align}\label{JabFpsiorder3}
[J^{ab}, F_{lmn} \psi^{dl} \psi_d^m \bar{\psi}^{cn}]= \frac{i}{2} \epsilon^{ca} F_{lmn}  \psi^{dl} \psi_d^m  \bar{\psi}^{bn} + \frac{i}{2} \epsilon^{cb}  F_{lmn} \psi^{dl} \psi_d^m  \bar{\psi}^{an}.
\end{align}
Then
\begin{align*}
[J^{ab}, Q^{21c}]&= -[J^{ab}, Q^c] = -[J^{ab}, p_l \psi^{cl}] - i F_{lmn} [J^{ab}, \langle \psi^{dl}\psi_d^m \bar{\psi}^{cn}\rangle] .
\end{align*}
Therefore we get from (\ref{QJmany1rep1form}) and (\ref{JabFpsiorder3}) that
\begin{align}\label{JabQ21cmanyparticle1rep}
[J^{ab}, Q^{21c}]&= \frac{i}{2}\big( \epsilon^{bc} p_l \psi^{al} +  \epsilon^{ac} p_l \psi^{bl}  -i \epsilon^{ca} F_{lmn} \langle \psi^{dl}\psi_d^m \bar{\psi}^{bn}\rangle-i \epsilon^{cb} F_{lmn} \langle \psi^{dl}\psi_d^m \bar{\psi}^{an}\rangle \big) \\
&= - \frac{i}{2} ( \epsilon^{cb} Q^a + \epsilon^{ca} Q^b) = i\epsilon^{c(a}Q^{|21|b)}\nonumber, 
\end{align}
as required in (\ref{odd-even}b). Further, we consider 
\begin{align*}
[J^{ab}, S^c] = -2 x_l [J^{ab}, \psi^{cl}] =i x_l ( \epsilon^{bc} \psi^{al} + \epsilon^{ac} \psi^{bl} )= \frac{i}{2}( \epsilon^{cb} S^a + \epsilon^{ca} S^b ) =  i\epsilon^{c(a}Q^{|11|b)},
\end{align*}
which coincides with the corresponding relation in (\ref{odd-even}b). The remaining relations can be proven similarly.
\end{proof}

\begin{lemma}\label{oddevencmanypartilcle1}Let $Q^{abc}$, $I^{ab}$ be as above. Then relations  (\ref{odd-even}c) hold.
\end{lemma}

\begin{proof}
Let us first consider  $[I^{11}, Q^{21a}]$. Using formulae (\ref{com1}), (\ref{com2}) we have
\begin{align}\label{blah1}
[ \psi_d^r \psi^{dr}, \psi^{bl} \psi_b^m \bar{\psi}^{an}] &= \psi^{bl} \psi_b^m [ \psi_d^r \psi^{dr}, \bar{\psi}^{an}] = - \psi^{bl} \psi_b^m \psi^{an}.
\end{align}
It follows that $F_{lmn} [ \psi_d^r \psi^{dr}, \psi^{bl} \psi_b^m \bar{\psi}^{an}]=0$ and hence
\begin{align}\label{blah2}
[I^{11}, Q^{21a}]=i [ \psi_d^r \psi^{dr}, Q^a]=i[\psi_d^r \psi^{dr}, p_l \psi^{al}]=0, 
\end{align}
as required for (\ref{odd-even}c). 

Let us now consider $[I^{22}, Q^{21a}]$. We have
\begin{align}\label{IQmany1rep1form}
[I^{22}, \psi^{al}]=i[\bar{\psi}^{dr} \bar{\psi}_{d}^r , \psi^{al}]= - i\bar{\psi}^{al}, 
\end{align}
and hence
\begin{align}\label{IQmany1rep2form}
[ \bar{\psi}^{dr} \bar{\psi}_d^r, \psi^{bl} \psi_b^m \bar{\psi}^{an}] &= - [ \psi^{bl} \psi_b^m , \bar{\psi}^{dr} \bar{\psi}_d^r ] \bar{\psi}^{an}\nonumber \\
&= ( \psi^{bl} [ \bar{\psi}^{dr} \bar{\psi}_d^r, \psi_b^m ] + [ \bar{\psi}^{dr} \bar{\psi}_d^r , \psi^{bl}]\psi_b^m ) \bar{\psi}^{an}\nonumber \\
&= - \psi^{bl} \bar{\psi}_b^m \bar{\psi}^{an} - \bar{\psi}^{bl} \psi_b^m \bar{\psi}^{an}.
\end{align}
By reordering terms in (\ref{IQmany1rep2form}) we obtain
\begin{align*}
[ \bar{\psi}^{dr} \bar{\psi}_d^r, \psi^{bl} \psi_b^m \bar{\psi}^{an}] &= ( \bar{\psi}_b^m \psi^{bl} + \delta^{lm} ) \bar{\psi}^{an} +  \bar{\psi}^{bl} ( \bar{\psi}^{an} \psi_b^m - \frac{1}{2} \delta^{a}_b \delta^{nm})\\
&= -\bar{\psi}_b^m   \bar{\psi}^{an} \psi^{bl}- \frac{1}{2} \bar{\psi}^{am} \delta^{ln} + \delta^{lm}  \bar{\psi}^{an} - \bar{\psi}_b^l  \bar{\psi}^{an} \psi^{bm}  - \frac{1}{2}  \bar{\psi}^{al} \delta^{nm}. 
\end{align*}
Therefore 
\begin{align}\label{IQmany1rep3form}
F_{lmn} [ \bar{\psi}^{dr} \bar{\psi}_d^r, \psi^{bl} \psi_b^m \bar{\psi}^{an}] = -2 F_{lmn}  \bar{\psi}_b^l  \bar{\psi}^{an} \psi^{bm}.
\end{align}
Note that $F_{lmn} \bar{\psi}_c^l \bar{\psi}^{an} \psi^{cm}=0$ if $c$ is fixed such that $c \neq a$. Hence (\ref{IQmany1rep3form}) can be rearranged as $-2 F_{lmn}  \bar{\psi}_a^l  \bar{\psi}^{am} \psi^{an}$ which is also equal to $-F_{lmn}  \bar{\psi}_b^l  \bar{\psi}^{bm} \psi^{an}$. Therefore 
\begin{align}\label{I22Q21amanyparticle1}
[I^{22}, Q^{21a}]&= -i [  \bar{\psi}^{dr} \bar{\psi}_d^r, Q^a]= -i \big( - p_l \bar{\psi}^{al} + i F_{lmn} ( - \bar{\psi}_b^l  \bar{\psi}^{bm} \psi^{an} + \frac{1}{2}  \bar{\psi}^{al} \delta^{nm} ) \big)= i \bar{Q}^a, 
\end{align}
as required for (\ref{odd-even}c). 

Further, let us consider $[I^{12}, Q^{21a}]= i [  \psi_d^r \bar{\psi}^{dr}, Q^a]$. Then by (\ref{commutatorforI12}) we have

\begin{align*}
[ \psi_d^r \bar{\psi}^{dr}, \psi^{bl} \psi_b^m \bar{\psi}^{an}] = \frac{1}{2} \psi^{bl} \psi_b^m \bar{\psi}^{an}. 
\end{align*}
Therefore, with the help of \eqref{QJmany1rep1form} we get
\begin{align}\label{I12Q21amanyparticle2}
[I^{12}, Q^{21a}] &= \frac{i}{2} \big( p_l \psi^{al} + i F_{lmn} ( \psi^{bl} \psi_b^m \bar{\psi}^{an} - \frac{1}{2} \psi^{al} \delta^{mn}) \big) = \frac{i}{2} Q^a, 
\end{align}
which matches with (\ref{odd-even}c). 

Let us now consider the generator $Q^{11a}$.  Firstly, it is immediate that $[I^{11}, Q^{11a}]=0$, as required. In addition, we have by (\ref{IQmany1rep1form}) that 
\begin{align*}
[I^{22}, Q^{11a}]= i [ \bar{\psi}^{dr} \bar{\psi}_d^r , S^a]=-2 i x_j [ \bar{\psi}^{dr}\bar{\psi}_d^r, \psi^{aj}]= -i \bar{S}^a, 
\end{align*}
and
\begin{align*}
[I^{12}, S^a]= -i [ \psi_d^r \bar{\psi}^{dr}, S^a]= i x_j \psi^{aj}=-\frac{i}{2} S^a, 
\end{align*}
as required for (\ref{odd-even}c). The remaining relations in (\ref{odd-even}c)  can be checked similarly. 
\end{proof}

Let $A_i$, $B_i$ ($i=1, 2, 3$) be operators. In the following theorem we will use the identity
\begin{align}\label{comm4}
\{A_1 A_2 A_3, B_1 B_2 B_3\} &= A_1 A_2 \{ A_3, B_1\} B_2 B_3 + A_1 A_2 B_1 B_2 \{ B_3, A_3\} - A_1 A_2 B_1 \{ B_2, A_3 \} B_3  - \nonumber\\
&- A_1 \{A_2, B_1\} B_2 B_3 A_3 - A_1 B_1 B_2 \{ B_3 , A_2\} A_3 + A_1 B_1 \{ B_2, A_2\} B_2 A_3 \nonumber\\
&+ \{A_1, B_1\} B_2 B_3 A_2 A_3 + B_1 B_2 \{ B_3, A_1\} A_2 A_3 - B_1 \{ B_2, A_1\} B_3 A_2 A_3.
\end{align}

We will use the following relations. We have by (\ref{com1}) and (\ref{com3})
\begin{align}\label{fortermBB'1}
\{ \psi^{ar}, \bar{\psi}_d^l \bar{\psi}^{dm} \psi_c^n \}&= \bar{\psi}_d^l [\bar{\psi}^{dm} \psi_c^n, \psi^{ar}] + \bar{\psi}^{dm} \psi_c^n \{ \bar{\psi}_d^l, \psi^{ar}\}= - \frac{1}{2} \bar{\psi}^{al} \psi_c^n \delta^{rm} - \frac{1}{2} \bar{\psi}^{am} \psi_c^n \delta^{rl},
\end{align}
and similarly,
\begin{align}\label{fortermBB'2}
\{\bar{\psi}_c^l, \psi^{br} \psi_b^j \bar{\psi}^{ak}\} = -\frac{1}{2} \psi_c^r \bar{\psi}^{ak} \delta^{jl} - \frac{1}{2} \psi_c^j \bar{\psi}^{ak} \delta^{rl}.
\end{align}

\begin{theorem}\label{hamiltonian1manyparticle}
For all $a, b =1, 2$ we have $\{Q^a, \bar{Q}_b  \}=-2 H \delta^a_b$, where the Hamiltonian $H$ is given by 
\begin{align}\label{hamiltonianformulamanyparticle1rep}
H= \frac{p^2}{4} - \frac{\partial_iF_{jlk}}{2} ( \psi^{bi} \psi_b^j  \bar{\psi}_d^l \bar{\psi}^{dk} - \psi_b^i \bar{\psi}^{bj} \delta^{lk} + \frac{1}{4} \delta^{ij} \delta^{lk} ) + \frac{1}{16} F_{ijk} F_{lmn} \delta^{nm} \delta^{jl} \delta^{ik}
\end{align}
with $p^2=\sum_{i=1}^N p_i^2$.
\end{theorem}

\begin{proof}
Let us consider $\{Q^{a}, \bar{Q}_c\}$, where
 \begin{align*}
Q^a = \overbrace{p_r \psi^{ar}}^{A} + \overbrace{ i F_{rjk} \langle \psi^{br} \psi_b^j  \bar{\psi}^{ak} \rangle}^{B}, \quad \bar{Q}_c = \overbrace{p_l  \bar{\psi}_c^l}^{A'} + \overbrace{i F_{lmn} \langle  \bar{\psi}_d^l  \bar{\psi}^{dm} \psi_c^n \rangle}^{B'}.
\end{align*}
We have
\begin{align*}
\{A, A' \}= - \frac{1}{2} \delta^{a}_c p^2.
\end{align*}
Further on, by (\ref{weylmanyparticle1rep2ndsuper}) we have
\begin{align*}
\{A, B'\} &= i \{ \psi^{ar} p_r , F_{lmn} \langle \bar{\psi}_d^l \bar{\psi}^{dm} \psi_c^n \rangle\}\\
&=  i \{ \psi^{ar} p_r , F_{lmn}  \bar{\psi}_d^l \bar{\psi}^{dm} \psi_c^n\} -\frac{i}{2} \delta^{nm} \{ p_r \psi^{ar}, F_{lmn} \bar{\psi}_c^l \}\\
&= i \psi^{ar} \bar{\psi}_d^l \bar{\psi}^{dm} \psi_c^n [p_r, F_{lmn}] + i\{ \psi^{ar}, \bar{\psi}_d^l \bar{\psi}^{dm} \psi_c^n \} F_{lmn} p_r -\\
&- \frac{i}{2} \delta^{nm} \psi^{ar} \bar{\psi}_c^l [p_r, F_{lmn}] + \frac{i}{4} \delta^{nm} \delta^{a}_c  F_{rmn} p_r.
\end{align*}
By (\ref{fortermBB'1}) we have $$F_{lmn} \{ \psi^{ar}, \bar{\psi}_d^l \bar{\psi}^{dm} \psi_c^{n} \}= - F_{lmn} \bar{\psi}^{al} \psi_c^n \delta^{rm}.$$
Therefore, 
\begin{align}\label{AB'manyparticle1}
\{A, B'\} =  i \psi^{ar} \bar{\psi}_d^l \bar{\psi}^{dm} \psi_c^n [p_r, F_{lmn}] - i  \bar{\psi}^{al} \psi_c^n F_{lnr} p_r- \frac{i}{2} \delta^{nm} \psi^{ar} \bar{\psi}_c^l [p_r, F_{lmn}] + \frac{i}{4} \delta^{nm} \delta^{a}_c  F_{rmn} p_r. 
\end{align}
Similarly, using (\ref{fortermBB'2}) we obtain
\begin{align}\label{BA'manyparticle1}
\{B, A'\} = i  \bar{\psi}_c^l \psi^{br} \psi_b^j  \bar{\psi}^{ak} [ p_l, F_{rjk}] - i \psi_c^r  \bar{\psi}^{ak} F_{rkj} p_j - \frac{i}{2} \delta^{jk}  \bar{\psi}_c^l \psi^{ar} [p_l, F_{rjk}] + \frac{i}{4} \delta^{jk} \delta^{a}_c F_{rjk} p_r. 
\end{align}
Note that $\bar{\psi}^{al} \psi_c^n F_{lnr} p_r + \psi_c^r \bar{\psi}^{ak} F_{rkj} p_j =\frac{1}{2} \delta^{ln} \delta^{ac} F_{lnr} p_r$. Then, after canceling out terms and simplifying we have
\begin{align}\label{(0)}
\{A, B'\}+ \{B, A'\} =  \partial_r F_{ljk} ( \psi^{ar} \bar{\psi}_d^l \bar{\psi}^{dk} \psi_c^j +  \bar{\psi}_c^l \psi^{br} \psi_b^j  \bar{\psi}^{ak}) + \frac{1}{4} \partial_r F_{lmn} \delta^{nm} \delta^{rl} \delta^{a}_c.
\end{align}
In particular, we note that using the symmetry of $F_{ljk}$ we have that
\begin{align}\label{(1)}
\partial_{r} F_{ljk} \psi^{ar} \bar{\psi}_d^l \bar{\psi}^{dk} \psi_c^j= \partial_r F_{ljk} ( \psi^{ar} \psi_c^j \bar{\psi}_d^l \bar{\psi}^{dk} + \psi^{ar} \bar{\psi}_c^k \delta^{lj}),
\end{align}
and
\begin{align}\label{(2)}
\partial_{r} F_{ljk}  \bar{\psi}_c^l \psi^{br} \psi_b^j  \bar{\psi}^{ak}= \partial_{r} F_{ljk} ( \psi^{br} \psi_b^j  \bar{\psi}_c^l  \bar{\psi}^{ak} - \psi_c^r  \bar{\psi}^{ak} \delta^{lj} ).
\end{align}
Note that if $a \neq c$, we have 
\begin{align}\label{(3)}
\psi^{ar} \bar{\psi}_c^k  = \psi_c^r \bar{\psi}^{ak} , \quad \text{and} \quad \psi^{ar} \psi_c^j = - \psi^{aj} \psi_c^r, \quad \bar{\psi}_c^l \bar{\psi}^{ak} = - \bar{\psi}_a^k \bar{\psi}^{cl}.
\end{align}
Using the symmetry $\partial_r F_{ljk} = \partial_l F_{rjk}$ and $F_{ljk} = F_{kjl}$ it follows from (\ref{(1)}), (\ref{(2)}) and (\ref{(3)}) that the sum of expressions in (\ref{(1)}) and (\ref{(2)}) vanishes if $ a \neq c$. Therefore we get from (\ref{(1)}), (\ref{(2)}), (\ref{(3)}) that
\begin{align}\label{(4)}
\partial_r F_{ljk} ( \psi^{ar} \bar{\psi}_d^l \bar{\psi}^{dk} \psi_c^j +  \bar{\psi}_c^l \psi^{br} \psi_b^j  \bar{\psi}^{ak}) & =  \partial_r F_{ljk} ( \psi^{ar} \psi_a^j \bar{\psi}_d^l \bar{\psi}^{dk}+ \psi^{br} \psi_b^j \bar{\psi}_a^l \bar{\psi}^{ak} - \psi_d^r \bar{\psi}^{dk} \delta^{lj}) \delta^a_c.
\end{align}
Note that 
\begin{align}\label{LAMBDA}
\psi^{ar} \psi_a^j = \psi^{\widehat{a}j} \psi_{\widehat{a}}^r, \quad \text{and} \quad \bar{\psi}^{ar} \bar{\psi}_a^j = \bar{\psi}^{\widehat{a}j} \bar{\psi}_{\widehat{a}}^r,
\end{align}
 here $\widehat{a}\neq a$. Therefore the right-hand side of (\ref{(4)}) equals 
\begin{align}\label{relationfor2ndrep}
\partial_r F_{ljk} ( \psi^{br} \psi_b^j \bar{\psi}_d^l \bar{\psi}^{dk} - \psi_d^r \bar{\psi}^{dk} \delta^{lj} ) \delta^a_c. 
\end{align}
Therefore in total expression (\ref{(0)}) becomes
\begin{align*}
\{A, B' \} + \{B, A'\}= \partial_r F_{ljk} ( \psi^{br} \psi_b^j \bar{\psi}_d^l \bar{\psi}^{dk} - \psi_d^r \bar{\psi}^{dk} \delta^{lj} + \frac{1}{4}\delta^{rl} \delta^{jk}  )\delta^a_c. 
\end{align*}

Finally, let us consider the term $\{B, B'\}$. We first show that 
\begin{align}\label{BB'noweyl}
C:=F_{rjk} F_{lmn} \{ \psi^{br} \psi_b^j \bar{\psi}^{ak}, \bar{\psi}_d^l \bar{\psi}^{dm} \psi_c^n \}=0.
\end{align}
By using (\ref{comm4}) we obtain
\begin{align*}
C&= F_{rjk} F_{lmn} \big( \psi^{br} \psi_b^j \bar{\psi}_d^l \bar{\psi}^{dm} \{ \psi_c^n, \bar{\psi}^{ak}\} - \psi^{br} \bar{\psi}^{dm} \psi_c^n \bar{\psi}^{ak} \{ \psi_b^j, \bar{\psi}_d^l \} \\
&+ \psi^{br} \bar{\psi}_d^l \psi_c^n \bar{\psi}^{ak} \{ \bar{\psi}^{dm}, \psi_b^j \} + \bar{\psi}^{dm} \psi_c^n \psi_b^j \bar{\psi}^{ak} \{ \psi^{br}, \bar{\psi}_d^l \} - \bar{\psi}_d^l \psi_c^n \psi_b^j \bar{\psi}^{ak} \{ \psi^{br}, \bar{\psi}^{dm}\} \big) \\
&= F_{rjk} F_{lmn}  \big(  \frac{1}{2} \delta^a_c \delta^{nk} \psi^{br} \psi_b^j \bar{\psi}_d^l \bar{\psi}^{dm}  + \frac{1}{2} \delta^{lj} \psi^{br}  \bar{\psi}_b^m \psi_c^n  \bar{\psi}^{ak} 
 + \frac{1}{2} \delta^{mj} \psi^{br}  \bar{\psi}_b^l \psi_c^n  \bar{\psi}^{ak} \\ &+ \frac{1}{2} \delta^{rl} \bar{\psi}^{dm} \psi_b^j \psi_c^n  \bar{\psi}^{ak} + \frac{1}{2} \delta^{rm}  \bar{\psi}^{bl} \psi_b^j \psi_c^n  \bar{\psi}^{ak} \big). 
\end{align*}
Then using the symmetry of $F_{lmn}$ under the swap of $l$ and $m$ we obtain
\begin{align*}
C= F_{rjk} F_{lmn} \big( \frac{1}{2} \delta^a_c \delta^{nk} \psi^{br} \psi_b^j \bar{\psi}_d^l \bar{\psi}^{dm} + \delta^{jl} \psi^{br} \bar{\psi}_b^m \psi_c^n \bar{\psi}^{ak} + \delta^{rl} \bar{\psi}^{bm} \psi_b^j \psi_c^n \bar{\psi}^{ak} \big).
\end{align*}
Note that by (\ref{extebracket1}),  (\ref{extebracket2}) we have
\begin{align}\label{forBB1'}
\psi^{br}  \bar{\psi}_b^m \psi_c^n  \bar{\psi}^{ak} = - \psi^{br} \psi_c^n  \bar{\psi}_b^m  \bar{\psi}^{ak} - \frac{1}{2} \psi_c^r  \bar{\psi}^{ak} \delta^{nm}, 
\end{align}
and 
\begin{align}\label{forBB'2}
 \bar{\psi}^{bm} \psi_b^j \psi_c^n  \bar{\psi}^{ak} &=- \psi_b^j \bar{\psi}^{bm} \psi_c^n \bar{\psi}^{ak} + \psi_c^n \bar{\psi}^{ak} \delta^{mj} \nonumber\\
 &= - \psi^{bj} \psi_c^n  \bar{\psi}_b^m  \bar{\psi}^{ak} - \frac{1}{2} \psi_c^j  \bar{\psi}^{ak} \delta^{nm} + \psi_c^n  \bar{\psi}^{ak} \delta^{mj}. 
\end{align} 
Further on by (\ref{WDVV}) we have $F_{rjk} F_{rmn} = F_{rnk} F_{rmj}$ and therefore some terms in the right-hand side of (\ref{forBB1'}), (\ref{forBB'2}) enter the relation 
\begin{align}\label{forBB'6}
F_{rjk} F_{rmn} \psi_c^n \bar{\psi}^{ak} \delta^{mj} = \frac{1}{2} F_{rjk} F_{jmn} \psi_c^r \bar{\psi}^{ak} \delta^{mn} + \frac{1}{2} F_{rjk} F_{rmn} \psi_c^j \bar{\psi}^{ak} \delta^{mn}.
\end{align}
Then by using (\ref{forBB1'})-(\ref{forBB'6}) and the symmetry of $F_{rjk}$ under the swap of $r$ and $j$ we obtain
\begin{align*}
C&= F_{rjk} F_{lmn} \big( \frac{1}{2} \delta^a_c \delta^{nk} \psi^{br} \psi_b^j \bar{\psi}_d^l \bar{\psi}^{dm} - \delta^{jl} \psi^{br} \psi_c^n \bar{\psi}_b^m \bar{\psi}^{ak} - \delta^{rl} \psi^{bj} \psi_c^n \bar{\psi}_b^m \bar{\psi}^{ak} \big)\\
&= F_{rjk} F_{lmn} \big( \frac{1}{2} \delta^a_c \delta^{nk} \psi^{br} \psi_b^j \bar{\psi}_d^l \bar{\psi}^{dm} - 2 \delta^{jl} \psi^{br} \psi_c^n \bar{\psi}_b^m \bar{\psi}^{ak}  \big).
\end{align*}
Note that for $c \neq a $ we have $C=0$, since  $F_{rjk} F_{lmn} \delta^{jl} \psi^{br} \psi_c^n \bar{\psi}_b^m \bar{\psi}^{ak}=0$ by using (\ref{WDVV}). Further on,  if $c=a$ then by using (\ref{WDVV}) we have
\begin{align}\label{finalc}
C= F_{rjk} F_{klm} \big( \frac{1}{2} \psi^{br} \psi_b^j \bar{\psi}_d^l \bar{\psi}^{dm} - 2  \psi^{br} \psi_a^j \bar{\psi}_b^l \bar{\psi}^{am}  \big).
\end{align}
Note that for $b \neq a$, $F_{rjk} \psi^{br} \psi_a^j=0$. Hence
\begin{align}\label{forBB'5}
F_{rjk} F_{klm} \psi^{br} \psi_a^j \bar{\psi}_b^l \bar{\psi}^{am}  = F_{rjk} F_{klm} \psi^{ar} \psi_a^j \bar{\psi}_a^l \bar{\psi}^{am} ,
\end{align}
which is equal to $\frac{1}{4} F_{rjk} F_{klm} \psi^{br} \psi_b^j \bar{\psi}_d^l \bar{\psi}^{dm}$ because of relations (\ref{LAMBDA}).  This proves that $C=0$.  Then the term $\{B, B'\}$ takes the following form:
\begin{align*}
\{B, B'\} &=  F_{rjk} F_{lmn} \big( \frac{1}{2} \delta^{nm} \{ \psi^{br} \psi_b^j \bar{\psi}^{ak}, \bar{\psi}_c^l \} + \frac{1}{2} \delta^{jk} \{ \bar{\psi}_d^l \bar{\psi}^{dm} \psi_c^n , \psi^{ar} \} - \frac{1}{4} \delta^{jk} \delta^{nm} \{ \psi^{ar}, \bar{\psi}_c^l \} \big).
\end{align*}
By using formulae (\ref{fortermBB'1}), (\ref{fortermBB'2}) and (\ref{WDVV}) we obtain
\begin{align*}
\{B, B'\}&= -\frac{1}{2} F_{rjk} F_{lmn} \big(  \psi_c^r \bar{\psi}^{ak}\delta^{nm}  \delta^{jl}  +   \bar{\psi}^{al} \psi_c^n \delta^{mr}\delta^{jk} - \frac{1}{4}\delta^{jk} \delta^{nm} \delta^{rl} \delta^a_c \big)\\
&= - \frac{1}{2} F_{rjk} F_{lmn} \delta^{nm} \delta^{jl} \{ \psi_c^r, \bar{\psi}^{ak} \} + \frac{1}{8} F_{rjk} F_{lmn} \delta^{jk} \delta^{nm} \delta^{rl} \delta^a_c \\
&= - \frac{1}{8} F_{rjk} F_{lmn} \delta^{nm} \delta^{jl} \delta^{rk} \delta^a_c. 
\end{align*}
Therefore, the statement follows. 
\end{proof}

\begin{lemma}\label{conformalinvarianceF} Let $T^{22}=H$ be given by Theorem \ref{hamiltonian1manyparticle}. Let $T^{11}=K$ and $T^{12}=-D$ be given by \eqref{K}, \eqref{D}. Then  relations (\ref{T}) hold.
\end{lemma}

\begin{proof} Firstly, we have that $[K, H]=\frac{1}{4} [x^2, p^2] = \frac{i}{2} \{x_r, p_r\}=-2iD$, as required. Moreover, since $H$ is homogeneous in $x$ of degree $-2$ it follows that $[H, D]=iH$ as required. Further on, $[K, D]=-\frac{1}{2} [x_k^2, x_j p_j] = i K$, 
which is the corresponding relation (\ref{T}). 
\end{proof}

\begin{lemma}\label{oddmanyparticle1rep}Let $Q^{abc}$, $I^{ab}$, $T^{ab}$, $J^{ab}$ be as above. Then  relations (\ref{odd}) hold. 
\end{lemma}

\begin{proof}
Firstly let us consider 
\begin{align*}
\{ Q^{21a}, Q^{11f}\}= - \{ Q^a, S^f\}. 
\end{align*}
Note that 
\begin{align*}
\{ p_r \psi^{ar},  x_l \psi^{fl}\}=-i \psi^{ar} \psi^{fr}=-i \epsilon^{a \widehat{a}} \psi_{\widehat{a}}^r \psi^{fr}, 
\end{align*}
where $\widehat{a}$ is complimentary to $a$. Note that we can assume now that $\widehat{a}=f$. Therefore
\begin{align*}
\{ p_r \psi^{ar},  x_l \psi^{fl}\}=-i\epsilon^{af} \psi_f^r \psi^{fr}=-\frac{i}{2} \epsilon^{af} \psi_d^r \psi^{dr}.
\end{align*}
Further,
\begin{align*}
F_{rjk} \{ \psi^{br} \psi_b^j \bar{\psi}^{ak},  x_l \psi^{fl}\} = \frac{1}{2}\epsilon^{af} x_k F_{krj} \psi_d^r \psi^{dj}.
\end{align*}
Therefore by formula \eqref{alphalambda}
\begin{align}
\{Q^{21a}, Q^{11f}\} =- i \epsilon^{af} \psi_d^r \psi^{dr}+ i \epsilon^{af} x_k F_{krj} \psi_d^r \psi^{dj}= 2(\alpha+1) \epsilon^{af} I^{11} \label{finalcon},
\end{align}
as required for the corresponding relation \eqref{odd}. 

Further on, consider $\{Q^{21a}, Q^{12b}\}=-\epsilon^{bd} \{Q^a, \bar{S}_d\}$. Now, by using formula (\ref{fortermBB'2}) we have
\begin{align*}
\{Q^a, \bar{S}_d\} &= -2\{ p_r \psi^{ar}, x_l \bar{\psi}_d^l \} -2i x_l F_{rjk} ( \{ \psi^{br} \psi_b^j \bar{\psi}^{ak}, \bar{\psi}_d^l \} - \frac{1}{2} \delta^{jk} \{ \psi^{ar}, \bar{\psi}_d^l \})\\
&= 2i \psi^{ar}  \bar{\psi}_d^r + x_r p_r \delta^a_d +2i x_j F_{jrk} \psi_d^r  \bar{\psi}^{ak} - \frac{i}{2} \delta^{jk} \delta^a_d x_r F_{r jk} \\
&=  2i \psi^{ar}  \bar{\psi}_d^r + x_r p_r \delta^a_d -2i (2\alpha+1)\psi_d^r  \bar{\psi}^{ar} + \frac{i \delta^a_d}{2}N(2\alpha+1).
\end{align*}
  Therefore 
\begin{align}\label{Jab1}
 \{Q^{21a}, Q^{12b}\}= -2i \psi^{ar}  \bar{\psi}^{br} + x_r p_r \epsilon^{ab} +2i (2\alpha+1)\psi^{br} \bar{\psi}^{ar} + \frac{i \epsilon^{ab} }{2}N(2\alpha+1).
\end{align}
Let us now note that 
\begin{align*}
I^{12} = -\frac{i}{2} [ \psi_a^j, \bar{\psi}^{aj}] = - i ( \psi^{2j} \bar{\psi}^{1j} - \psi^{1j} \bar{\psi}^{2j} - \frac{N}{2}). 
\end{align*}
Hence the right-hand side of (\ref{odd}) for $\{Q^{21a}, Q^{12b}\}$ is 
\begin{align}\label{Jab2}
 x_r p_r \epsilon^{ab} - \frac{iN}{2} \epsilon^{ab} + 4i \alpha \psi^{(ar}\bar{\psi}^{br)} - 2i (1+\alpha)\epsilon^{ab} ( \psi^{2j} \bar{\psi}^{1j} - \psi^{1j} \bar{\psi}^{2j} - \frac{N}{2}).
\end{align}
By considering various values of $a, b \in \{1, 2\}$, expression (\ref{Jab2}) takes the form
\begin{align}\label{Jab2new}
 x_r p_r \epsilon^{ab} + \frac{i \epsilon^{ab} }{2} N (2\alpha+1) - 2i  \psi^{ar} \bar{\psi}^{br} + 2i (2\alpha +1) \psi^{br} \bar{\psi}^{ar},
\end{align}
which is equal to (\ref{Jab1}) as required, so the corresponding relation \eqref{odd} follows.  

Further on, let us consider relation $\{Q^{21a}, Q^{21b}\}=\{Q^a, Q^b \}$. By using (\ref{extebracket1}) and (\ref{extebracket2}) we have
\begin{align*}
\{Q^a, Q^c\}&=i \{ p_r \psi^{ar}, F_{lmn} \psi^{dl} \psi_{d}^m \bar{\psi}^{cn}\} + i \{ p_l \psi^{cl},  F_{rjk} \psi^{br} \psi_b^j  \bar{\psi}^{ak} \} - \\
&-F_{lmn}F_{rjk} \{ \langle \psi^{dl} \psi_{d}^m \bar{\psi}^{cn}\rangle, \langle \psi^{br} \psi_b^j  \bar{\psi}^{ak} \rangle\}.
\end{align*}
Note that by (\ref{com2}), (\ref{com3}) we have
\begin{align*}
\{ p_r \psi^{ar}, F_{lmn} \psi^{dl} \psi_{d}^m \bar{\psi}^{cn}\} &= \psi^{ar} \psi^{dl} \psi_{d}^m \bar{\psi}^{cn}[ p_r, F_{lmn}] + \{ \psi^{ar}, \psi^{dl} \psi_d^m \bar{\psi}^{cn} \} F_{lmn} p_r \\
&=-i \psi^{ar} \psi^{dl} \psi_{d}^m \bar{\psi}^{cn} \partial_r F_{lmn} + \{ \psi^{ar}, \bar{\psi}^{cn}\} \psi^{dl} \psi_d^m F_{lmn} p_r\\
&= -i \psi^{ar} \psi^{dl} \psi_{d}^m \bar{\psi}^{cn} \partial_r F_{lmn} - \frac{1}{2} \epsilon^{ca}  \psi^{dl} \psi_d^m F_{lmr} p_r.
\end{align*} 
Note also that $\psi^{ar} \psi^{al} \partial_r F_{lmn}=0$ using the symmetry of $\partial_r F_{lmn}$ under the swap of $r$ and $l$.  Then $\psi^{ar} \psi^{dl} \psi_{d}^m \bar{\psi}^{cn} \partial_r F_{lmn} =0$ and hence 
\begin{align}\label{QQmanypartilce1repform1}
\{ p_r \psi^{ar}, F_{lmn} \psi^{dl} \psi_{d}^m \bar{\psi}^{cn}\} = - \frac{1}{2} \epsilon^{ca} F_{lmr} p_r \psi^{dl} \psi_d^m.
\end{align}
Similarly,
\begin{align}\label{QQmanypartilce1repform2}
\{ p_l \psi^{cl},  F_{rjk} \psi^{br} \psi_b^j  \bar{\psi}^{ak} \} &= -i \psi^{cl} \psi^{br} \psi_b^j  \bar{\psi}^{ak} \partial_l F_{rjk} - \frac{1}{2} \epsilon^{ac} F_{rjk} p_k \psi^{br} \psi_b^j \nonumber\\
&= - \frac{1}{2} \epsilon^{ac} F_{rjk} p_k \psi^{br} \psi_b^j.
\end{align}
Note that terms in (\ref{QQmanypartilce1repform1}) and (\ref{QQmanypartilce1repform2}) cancel. Further,  we have
\begin{align}
F_{lmn}F_{rjk} \{ \langle \psi^{dl} \psi_{d}^m \bar{\psi}^{cn}\rangle, \langle \psi^{br} \psi_b^j  \bar{\psi}^{ak} \rangle\}&= F_{lmn}F_{rjk} \{  \psi^{dl} \psi_{d}^m \bar{\psi}^{cn},  \psi^{br} \psi_b^j  \bar{\psi}^{ak} \}\nonumber \\
&+ \frac{1}{4} \epsilon^{ca} F_{lmr}F_{rjj} \psi^{dl} \psi_d^m + \frac{1}{4} \epsilon^{ac} F_{rjk}F_{kmm}\psi^{br} \psi_b^j \label{QQmanyparticle1repform3} \\
&=F_{lmn}F_{rjk} \{  \psi^{dl} \psi_{d}^m \bar{\psi}^{cn},  \psi^{br} \psi_b^j  \bar{\psi}^{ak} \} \nonumber,
\end{align} 
since the last two terms in (\ref{QQmanyparticle1repform3}) cancel. Note that by (\ref{comm4}) we have
\begin{align*}
\{  \psi^{dl} \psi_{d}^m \bar{\psi}^{cn},  \psi^{br} \psi_b^j  \bar{\psi}^{ak} \}&= \psi^{dl} \psi_d^m\big( \psi_b^j \bar{\psi}^{ak} \{ \bar{\psi}^{cn}, \psi^{br} \} - \psi^{br} \bar{\psi}^{ak} \{ \psi_b^j , \bar{\psi}^{cn} \} \big)\\
&+ \psi^{br} \psi_b^j \big(\psi_d^m \bar{\psi}^{cn} \{ \bar{\psi}^{ak}, \psi^{dl}\} - \psi^{dl} \bar{\psi}^{cn} \{ \bar{\psi}^{ak}, \psi_d^m \} \big) \\
&= - \frac{1}{2} \psi^{dl} \psi_d^m \big( \psi^{cj} \delta^{nr} + \psi^{cr} \delta^{jn} \big) \bar{\psi}^{ak} - \frac{1}{2} \psi^{br} \psi_b^j \big( \psi^{al} \delta^{km} + \psi^{am} \delta^{kl} \big) \bar{\psi}^{cn}.
\end{align*}
Therefore using the symmetry of $F_{rjk}$ under the swap of $j$ and $r$, and that of $F_{lmn}$ under the swap of $l$ and $m$ we obtain
\begin{align}\label{QQmanyparticle1repform4} 
F_{lmn}F_{rjk}\{  \psi^{dl} \psi_{d}^m \bar{\psi}^{cn},  \psi^{br} \psi_b^j  \bar{\psi}^{ak} \}=-F_{lmn}F_{rjk} \big( \psi^{dl} \psi_d^m  \psi^{cj} \delta^{nr}+  \psi^{br} \psi_b^j \psi^{al} \delta^{km}\big). 
\end{align}
Further, note that for any $b \in \{1, 2\}$ we have by using (\ref{WDVV}) that $F_{lmr} F_{rjk} \psi^{dl} \psi_d^m \psi^{bj} =0$. Hence the right-hand side of (\ref{QQmanyparticle1repform4}) vanishes. Therefore it follows that $$F_{lmn}F_{rjk} \{ \langle \psi^{dl} \psi_{d}^m \bar{\psi}^{cn}\rangle, \langle \psi^{br} \psi_b^j  \bar{\psi}^{ak} \rangle\}=0$$ and hence that $\{Q^a, Q^b\}=0$ as required. 

Further on it is easy to see that $ \{ Q^{11a}, Q^{11b}\} =\{Q^{12a}, Q^{12b}\}=0$. By Theorem \ref{hamiltonian1manyparticle} we have $\{ Q^{21a}, Q^{22b}\} = -2 H \epsilon^{ba}$. The remaining relations \eqref{odd} can be shown in a similar way.
\end{proof}

\begin{lemma}\label{oddevenamanyparticle1}Let $T^{ab}$, $Q^{abc}$ be as above. Then relations (\ref{odd-even}a) hold. 
\end{lemma}

\begin{proof}
Firstly, it is easy to see that $[T^{11}, Q^{21a}]=-[K, Q^a]= -2i x_r \psi^{ar} = iS^a$, and $[T^{11}, Q^{11a}]=[K,S^a]=0$, and  $[T^{12}, Q^{11a}]=-[D, S^a]= -\frac{i}{2}Q^{11a}$. Moreover, we have
 $[T^{12}, Q^{21a}]=[D, Q^a]= \frac{i}{2} Q^{21a}$ as $Q^a$ is homogeneous in $x$ of degree $-1$. This gives relations  
(\ref{odd-even}a) for commutators between $K, D$ and $Q^a, S^a$. 

Further, we have
\begin{align*}
[ \psi^{br} \psi_b^j \bar{\psi}_d^l \bar{\psi}^{dk} , \psi^{am}]=\frac{1}{2} \psi^{br} \psi_b^j ( \bar{\psi}^{al} \delta^{km} + \bar{\psi}^{ak} \delta^{lm}),
\end{align*}
therefore 
\begin{align}\label{2)}
\partial_r F_{jlk} [ \psi^{br} \psi_b^j \bar{\psi}_d^l \bar{\psi}^{dk}, \psi^{am}]= \partial_r F_{jlm} \psi^{br} \psi_b^j \bar{\psi}^{al}. 
\end{align}
Note also that 
\begin{align}\label{3)}
\partial_r F_{jlk} [\psi_b^r \bar{\psi}^{bj} \delta^{lk}, \psi^{am}]= \frac{1}{2} \partial_{r} F_{lmk} \psi^{ar} \delta^{lk}. 
\end{align}
Hence we get from (\ref{2)}) and (\ref{3)}) that 
\begin{align}\label{blah}
\partial_r F_{jlk} [ \psi^{br} \psi_b^j \bar{\psi}_d^l \bar{\psi}^{dk}- \psi_b^r \bar{\psi}^{bj} \delta^{lk} , \psi^{am}] &= \partial_r F_{jlm} \psi^{br} \psi_b^j \bar{\psi}^{al} - \frac{1}{2} \partial_{r} F_{lmk} \psi^{ar} \delta^{lk} \\
&= \partial_m F_{rjl} \langle \psi^{br} \psi_b^j \bar{\psi}^{al} \rangle \nonumber,
\end{align}
in view of (\ref{weylmanyparticle1rep}). Therefore 
\begin{align}\label{HS^amanyparticle1rep}
[ H, S^a]&= i p_r \psi^{ar} + x_m \partial_m F_{rjl} \langle \psi^{br} \psi_b^j \bar{\psi}^{al}\rangle\\
&= i p_r \psi^{ar} - F_{rjl} \langle \psi^{br} \psi_b^j \bar{\psi}^{al}\rangle \nonumber\\
&= i Q^a\nonumber, 
\end{align}
as required for (\ref{odd-even}a). Further on, by Theorem \ref{hamiltonian1manyparticle} we have
$T^{22}=H= - \frac{1}{2} \{Q^a, \bar{Q}_a\}.$ Since $(Q^a)^2=0$ we get that $[H, Q^a]=0$ as required. 
The remaining relations (\ref{odd-even}a) can be shown in a similar way. 
\end{proof}

\begin{lemma}\label{thezerothrelationslemma}Let $T^{ab}$, $I^{ab}$, $J^{ab}$ be as above. Then relations (\ref{thezerorelations}) hold. 
\end{lemma}

\begin{proof}
Let us firstly  consider $[I^{ab}, J^{cd}]$. We have by (\ref{com1}) and (\ref{forproofofIS}) that $$[\psi_a^j \psi^{aj}, \psi^{ck} \bar{\psi}^{dk}]= \psi^{dk} \psi^{ck}.$$ Therefore
\begin{align*}
[I^{11}, J^{cd}]&=2[\psi_a^j \psi^{aj}, \psi^{(ck} \bar{\psi}^{dk)}]=0,
\end{align*}
as required. Further,  we have by (\ref{com1}), (\ref{com2}) that
\begin{align*}
[[\psi_a^j ,  \bar{\psi}^{aj}], \psi^{ck}  \bar{\psi}^{dk}] &= 2 [ \psi_a^j  \bar{\psi}^{aj}, \psi^{ck}  \bar{\psi}^{dk}] \\
&= 2 (\psi_a^j [  \bar{\psi}^{aj}, \psi^{ck}  \bar{\psi}^{dk}]+ [ \psi_a^j , \psi^{ck}  \bar{\psi}^{dk}]  \bar{\psi}^{aj} )\\
&= 2 ( \psi_a^j  \bar{\psi}^{dk} \{ \psi^{ck},  \bar{\psi}^{aj}\} - \psi^{ck}  \bar{\psi}^{aj} \{  \bar{\psi}^{dk}, \psi_a^j \}) =0.
\end{align*}
Therefore, 
\begin{align*}
[I^{12}, J^{cd}]=[ [\psi_a^j ,  \bar{\psi}^{aj}],  \psi^{(ck} \bar{\psi}^{dk)}]=0,
\end{align*}
which is the corresponding relation (\ref{thezerorelations}). In addition we have by (\ref{com1}) and (\ref{forproofofIS}) that 
\begin{align*}
[  \bar{\psi}^{aj} \bar{\psi}_a^j , \psi^{ck} \bar{\psi}^{dk}] = - \bar{\psi}^{ck} \bar{\psi}^{dk}. 
\end{align*}
Therefore,
\begin{align*}
[I^{22}, J^{cd}]= -2 [ \bar{\psi}^{aj} \bar{\psi}_a^j ,  \psi^{(ck} \bar{\psi}^{dk)}]=0, 
\end{align*}
as required. 

Let us now consider relations $[I^{ab}, T^{cd}]$, $(a, b, c, d=1, 2$).  It is easy to see that for $T^{12}=-D$ and $T^{11}=K$ relations (\ref{thezerorelations}) hold.  Further, we have $T^{22}=H= -\frac{1}{2} \{ Q^c, \bar{Q}_c\}$. Then by (\ref{com1}) we obtain
\begin{align*}
[I^{ab}, H]&= - \frac{1}{2} ( [I^{ab}, Q^c \bar{Q}_c ] + [I^{ab}, \bar{Q}_c Q^c] ) \\
&= - \frac{1}{2}  ( Q^c [ I^{ab}, \bar{Q}_c] + [I^{ab}, Q^c]\bar{Q}_c + \bar{Q}_c [ I^{ab}, Q^c] + [I^{ab},  \bar{Q}_c] Q^c) \\
&=- \frac{1}{2}  (- Q_{\widehat{c}} [ I^{ab}, \bar{Q}^{\widehat{c}}] + [I^{ab}, Q^c]\bar{Q}_c + \bar{Q}_c [ I^{ab}, Q^c] - [I^{ab},  \bar{Q}^{\widehat{c}}] Q_{\widehat{c}}),
\end{align*}
where $\widehat{c}$ is complimentary to $c$. 
%
Then by Lemma \ref{oddevencmanypartilcle1} we have $$[I^{ab}, Q^c]= -[I^{ab}, Q^{21c}]=-\frac{i}{2} ( \epsilon^{ 1a} Q^{2bc} + \epsilon^{1b} Q^{2ac}) \quad \text{and} \quad  [ I^{ab}, \bar{Q}^c] =-\frac{ i}{2} (\epsilon^{2a} Q^{2bc} + \epsilon^{2b} Q^{2ac}).$$ 
Therefore by considering various values of $a,  b \in \{1, 2\}$  and by using Lemma \ref{oddmanyparticle1rep} and Theorem \ref{hamiltonian1manyparticle} we obtain the following:
\begin{align*}
[I^{11}, H]&= \frac{i}{2}( Q_{\widehat{c}} Q^{\widehat{c}} + Q^{\widehat{c}} Q_{\widehat{c}})=0,\\
[I^{22}, H]&= \frac{i}{2} ( \bar{Q}^c \bar{Q}_c + \bar{Q}_c \bar{Q}^c )=0,\\
[I^{12}, H]&=  \frac{i}{2} ( Q_{\widehat{c}} \bar{Q}^{\widehat{c}} + Q^c \bar{Q}_c + \bar{Q}_c Q^c + \bar{Q}^{\widehat{c}}Q_ {\widehat{c}})=0,
\end{align*}
which are the corresponding relations (\ref{thezerorelations}).

Similarly we have
\begin{align*}
[J^{ab}, H]=- \frac{1}{2}  (- Q_{\widehat{c}} [ J^{ab}, \bar{Q}^{\widehat{c}}] + [J^{ab}, Q^c]\bar{Q}_c + \bar{Q}_c [ J^{ab}, Q^c] - [J^{ab},  \bar{Q}^{\widehat{c}}] Q_{\widehat{c}}).
\end{align*}
By Lemma \ref{oddevenbmanyparticle1} we have
\begin{align*}
[J^{ab}, Q^c]=\frac{i}{2} (\epsilon^{ca} Q^b + \epsilon^{cb} Q^a) \quad \text{and} \quad [J^{ab}, \bar{Q}^c] = \frac{i}{2} ( \epsilon^{ca} \bar{Q}^b + \epsilon^{cb} \bar{Q}^a). 
\end{align*}
Therefore by considering various values of $a, b \in \{1, 2\}$ we obtain:
\begin{align}\label{J11H}
[J^{11}, H] =-\frac{i}{2} (-\epsilon^{\widehat{c}1} Q_{\widehat{c}} \bar{Q}^1 + \epsilon^{c1} Q^1 \bar{Q}_c + \epsilon^{c1} \bar{Q}_c Q^1 -\epsilon^{\widehat{c}1} \bar{Q}^1Q_{\widehat{c}}),
\end{align}
\begin{align*}
[J^{12}, H] =& - \frac{i}{4} (-\epsilon^{\widehat{c}1} Q_{\widehat{c}} \bar{Q}^2 - \epsilon^{\widehat{c}2} Q_{\widehat{c}} \bar{Q}^1 + \epsilon^{c1} Q^2 \bar{Q}_c + \epsilon^{c2} Q^1 \bar{Q}_c \nonumber\\
&+ \epsilon^{c1} \bar{Q}_c Q^2 + \epsilon^{c2} \bar{Q}_c Q^1 - \epsilon^{\widehat{c}1} \bar{Q}^2 Q_{\widehat{c}} -\epsilon^{\widehat{c}2} \bar{Q}^1 Q_{\widehat{c}} ),
\end{align*}
\begin{align}\label{J22H}
[J^{22}, H] =-\frac{i}{2} (-\epsilon^{\widehat{c}2} Q_{\widehat{c}} \bar{Q}^2 + \epsilon^{c2} Q^2 \bar{Q}_c + \epsilon^{c2} \bar{Q}_c Q^2 -\epsilon^{\widehat{c}2} \bar{Q}^2Q_{\widehat{c}}).
\end{align}
Then by considering various values of $c \in \{1, 2\}$ in (\ref{J11H})--(\ref{J22H}) and by using Lemma \ref{oddmanyparticle1rep} and Theorem \ref{hamiltonian1manyparticle} we obtain that
\begin{align*}
[J^{11}, H]=[J^{12}, H]=[J^{22}, H]=0, 
\end{align*}
as required for (\ref{thezerorelations}). 
\end{proof}

\section{The second representation} 
\label{2nd rep section}

Let now the supercharges be of the form
\begin{align}\label{newQ}
Q^a& = p_r \psi^{ar} + i F_{rjk} \psi^{br} \psi_b^j  \bar{\psi}^{ak} ,
\end{align}
\begin{align}\label{newQbar}
\bar{Q}_c &= p_l  \bar{\psi}_c^l + i F_{lmn}  \bar{\psi}_d^l  \bar{\psi}^{dm} \psi_c^n,
\end{align}
so we do not have anti-symmetrisation in the cubic fermionic terms. 
Let generators $K$,  $I^{ab}$, $J^{ab}$, and $S^a, \bar S_a$ be given by formulas
\eqref{K}, \eqref{I}, \eqref{J},  \eqref{S} same as in the first representation, while the generator $D$ is now given by
\begin{align}
\label{newD}
D= -\frac{1}{2}x_j  p_j+\frac{i}{2} (\alpha+1)N. 
\end{align}

\begin{theorem}\label{hamiltonianmanyparticle2} For all $a, b = 1, 2$ we have $\{Q^a, \bar{Q}_b\}=-2 H \delta^a_b$, where the Hamiltonian $H$ is 
\begin{align}
\label{newH}
H= \frac{p^2}{4} - \frac{\partial_rF_{jlk}}{2} ( \psi^{br} \psi_b^j  \bar{\psi}_d^l \bar{\psi}^{dk} - \psi_b^r \bar{\psi}^{bj} \delta^{lk} ) + \frac{i}{4} \delta^{nm} F_{rmn} p_r. 
\end{align}
\end{theorem}

\begin{proof}
 Let us denote terms in (\ref{newQ}), (\ref{newQbar}) as follows:
\begin{align*}
Q^a = \overbrace{p_r \psi^{ar}}^{A} +\overbrace{ i F_{rjk} \psi^{br} \psi_b^j  \bar{\psi}^{ak}}^{B}, \quad 
\bar{Q}_c = \overbrace{p_l  \bar{\psi}_c^l}^{A' }+ \overbrace{ i F_{lmn}  \bar{\psi}_d^l  \bar{\psi}^{dm} \psi_c^n}^{B'}.
\end{align*}
Then, analogues of relations (\ref{AB'manyparticle1}), (\ref{BA'manyparticle1}) are
\begin{align}\label{AB'manyparticle2}
\{A, B'\}&= i \psi^{ar} \bar{\psi}_d^l \bar{\psi}^{dm} \psi_c^n [p_r, F_{lmn}] - i  \bar{\psi}^{al} \psi_c^n F_{lnr} p_r, 
\end{align}
and 
\begin{align}\label{BA'manyparticle2}
 \{B,A'\}&= i  \bar{\psi}_c^l \psi^{br} \psi_b^j  \bar{\psi}^{ak} [ p_l, F_{rjk}] - i \psi_c^r  \bar{\psi}^{ak} F_{rkj} p_j,
\end{align}
respectively. Then using (\ref{AB'manyparticle2}) and (\ref{BA'manyparticle2}) an analogue of equality (\ref{(0)}) is (cf. \eqref{relationfor2ndrep})
\begin{align*}
\{A, B'\}+ \{B, A'\} &=  \partial_r F_{ljk} ( \psi^{ar} \bar{\psi}_d^l \bar{\psi}^{dk} \psi_c^j +  \bar{\psi}_c^l \psi^{br} \psi_b^j  \bar{\psi}^{ak}) - \frac{i}{2} \delta^{nl}F_{lnr}p_r \delta^{a}_c \\
&=  \partial_rF_{jlk} ( \psi^{br} \psi_b^j  \bar{\psi}_d^l \bar{\psi}^{dk} - \psi_b^r \bar{\psi}^{bj} \delta^{lk} )\delta^a_c-\frac{i}{2} \delta^{nl}F_{lnr}p_r \delta^{a}_c.
\end{align*}
Further on we have $\{B, B'\}=0$ (cf. (\ref{BB'noweyl})). Therefore in total, we get that 
\begin{align}
\{Q^a, \bar{Q}_c \}&=-
\frac{p^2}{2}\delta^a_c + \{A, B'\} + \{B, A'\}\nonumber \\
&= -\frac{p^2}{2} \delta^a_c+ \partial_rF_{jlk} ( \psi^{br} \psi_b^j  \bar{\psi}_d^l \bar{\psi}^{dk} - \psi_b^r \bar{\psi}^{bj} \delta^{lk} )\delta^a_c - \frac{i}{2} \delta^{nm} F_{rmn} p_r\delta^a_c \label{newQQbar},
\end{align}
and hence the statement follows. 
\end{proof}

\begin{lemma}\label{conformalinvarianceF2} Let $T^{ab}$ be given by \eqref{newD}, \eqref{newH} and \eqref{K}. Then relations (\ref{T}) hold.
\end{lemma}

\begin{proof} Firstly, we have  that
\begin{align*}
[K, H]&=\frac{1}{4} [x^2, p^2] + \frac{i}{4}\delta^{nm} F_{rmn} [x^2, p_r] = \frac{i}{2} \{x_r, p_r\} + \frac{N}{2} (2\alpha+1) =-2iD, 
\end{align*}
as required. Moreover we have $[F_{rmn} p_r, x_j p_j]=-iF_{rmn}p_r + ix_j \partial_j F_{rmn} p_r=-2iF_{rmn}p_r$. Then it is easy to see that $[H, D]=iH$, as required. Further on, $[K, D]=-\frac{1}{2} [x^2, x_j p_j] = i K$, 
which is the corresponding relation (\ref{T}). 
\end{proof}

We note that since $I$ and $J$ keep the same form as in the first representation, the statement of the Lemmas \ref{Jmanyparticle1}, \ref{Imanyparticle1} hold. 

\begin{lemma} Let  $Q^{abc}$, $I^{ab}$, $J^{ab}$ be given by \eqref{newQ}, \eqref{newQbar}, \eqref{S}, \eqref{I}, \eqref{J}. Then relations (\ref{odd-even}b), (\ref{odd-even}c) hold. 
\end{lemma}

\begin{proof}
Relations (\ref{odd-even}b),(\ref{odd-even}c) are easy to verify by an adaptation of the proof of Lemmas \ref{oddevenbmanyparticle1} and \ref{oddevencmanypartilcle1} respectively. Indeed let us first consider relations (\ref{odd-even}b) for $[J^{ab}, Q^{21c}]$, which  now takes the form (cf. (\ref{JabQ21cmanyparticle1rep}))
\begin{align*}
[J^{ab}, Q^{21c}]&= \frac{i}{2}\big( \epsilon^{bc} p_l \psi^{al} +  \epsilon^{ac} p_l \psi^{bl}  -i \epsilon^{ca} F_{lmn}  \psi^{dl}\psi_d^m \bar{\psi}^{bn}-i \epsilon^{cb} F_{lmn} \psi^{dl}\psi_d^m \bar{\psi}^{an} \big) \\
&= - \frac{i}{2} ( \epsilon^{cb} Q^a + \epsilon^{ca} Q^b) = i\epsilon^{c(a}Q^{|21|b)}\nonumber, 
\end{align*}
as required for (\ref{odd-even}b).

Further on, let us consider relations (\ref{odd-even}c) for $[I^{ab}, Q^{21c}]$. Expression (\ref{I22Q21amanyparticle1}) now takes the form
\begin{align*}
[I^{22}, Q^{21a}]&= -i [  \bar{\psi}^{dr} \bar{\psi}_d^r, Q^a]= i \big( p_l \bar{\psi}^{al} + i F_{lmn}  \bar{\psi}_b^l  \bar{\psi}^{bm} \psi^{an}  \big)= i \bar{Q}^a, 
\end{align*}
as required. The analogue of (\ref{I12Q21amanyparticle2}) is
\begin{align*}
[I^{12}, Q^{21a}] &= \frac{i}{2} \big( p_l \psi^{al} + i F_{lmn}  \psi^{bl} \psi_b^m \bar{\psi}^{an}  \big) = \frac{i}{2} Q^a, 
\end{align*}
which matches  (\ref{odd-even}c). Finally,  it is easy to see that $[I^{11}, Q^{21a}]=0$ (cf. (\ref{blah1}), (\ref{blah2}) in Lemma \ref{oddevencmanypartilcle1}).  Relations (\ref{odd-even}) for $S^a$ take the same form as in Lemmas \ref{oddevenbmanyparticle1} and \ref{oddevencmanypartilcle1}. The remaining relations can be checked in a similar way.
\end{proof}
 
\begin{lemma}\label{oddmanyparticle2rep} Let $Q^{abc}$, $I^{ab}$, $J^{ab}$,  $T^{ab}$  be given by formulas \eqref{newQ}, \eqref{newQbar}, \eqref{S}, \eqref{I}, \eqref{J}, \eqref{K},  \eqref{newD},   \eqref{newH}. Then relations (\ref{odd}) hold. 
\end{lemma}

\begin{proof}
We first note that by Theorem \ref{hamiltonianmanyparticle2} we have $\{Q^a, \bar{Q}^c \}= \epsilon^{cb} \{ Q^a, \bar{Q}_b\} = -2 H \epsilon^{ca}$ which is the corresponding relation (\ref{odd}). 
The anticommutator $\{Q^{21a}, Q^{21b}\}$ vanishes since the terms (\ref{QQmanypartilce1repform1}), (\ref{QQmanypartilce1repform2}) cancel each other and the right-hand side of (\ref{QQmanyparticle1repform4}) vanishes. Further on it is immediate that $\{Q^{21a}, Q^{11f}\}$ is the same as in the first representation. Similarly for  $\{Q^{22a}, Q^{22b}\}$, $\{Q^{22a}, Q^{12f} \}$. Note also that  $\{Q^{11a}, Q^{11b}\}$, $\{Q^{12a}, Q^{12b}\}$, $\{Q^{11a}, Q^{12b}\}$ take the same form as in Lemma \ref{oddmanyparticle1rep}. 

Further on, let us consider $\{Q^{21a}, Q^{12b}\}$. The left-hand side of (\ref{odd}) now takes the form (cf.(\ref{Jab1}) and the change in the generator $D$)
\begin{align}\label{Jab21}
 \{Q^{21a}, Q^{12b}\}= -2i \psi^{ar}  \bar{\psi}^{br} + x_r p_r \epsilon^{ab} +2i (1+2\alpha)\psi^{br} \bar{\psi}^{ar},
\end{align}
and the right-hand side of (\ref{odd}) becomes (cf. (\ref{Jab2new})) 
\begin{align*}
\{Q^{21a}, Q^{12b}\} &=  x_r p_r \epsilon^{ab} + 4i \alpha \psi^{(ar}\bar{\psi}^{br)} -2i(1+\alpha) \epsilon^{ab} ( \psi^{2r} \bar{\psi}^{1r} - \psi^{1r} \bar{\psi}^{2r} )\nonumber \\
&= -2i \psi^{ar}  \bar{\psi}^{br} + x_r p_r \epsilon^{ab} +2i (1+2\alpha)\psi^{br} \bar{\psi}^{ar},
\end{align*}
which is equal to (\ref{Jab21}) as required.  The remaining relations can be checked similarly. 
\end{proof}

\begin{lemma} Let $T^{ab}$ and $Q^{abc}$ be given by \eqref{K}, \eqref{newD}, \eqref{newH}, \eqref{newQ},  \eqref{newQbar}, \eqref{S}. Then relations (\ref{odd-even}a) hold. 
\end{lemma}

\begin{proof}
Firstly, it is easy to see that $[T^{11}, Q^{21a}]=-[K, Q^a]= -2i x_r \psi^{ar} = iS^a$, and $[T^{11}, Q^{11a}]=[K,S^a]=0$, and  $[T^{12}, Q^{11a}]=-[D, S^a]= -\frac{i}{2}Q^{11a}$. Moreover, we have
 $[T^{12}, Q^{21a}]= \frac{i}{2} Q^{21a}$ as $Q^a$ is homogeneous in $x$ of degree $-1$. 

Let us recall that from the proof of Lemma \ref{oddevenamanyparticle1} (formula (\ref{blah})) we have 
\begin{align*}
\partial_r F_{jlk} [ \psi^{br} \psi_b^j \bar{\psi}_d^l \bar{\psi}^{dk}- \psi_b^r \bar{\psi}^{bj} \delta^{lk}, \psi^{am}]=\delta^{km} \partial_k F_{rjl} ( \psi^{br} \psi_b^j \bar{\psi}^{al}- \frac{1}{2} \delta^{jl} \psi^{ar}). 
\end{align*}
Therefore an analogue of (\ref{HS^amanyparticle1rep}) takes the form
\begin{align*}
[H, S^a]&= -\frac{1}{2}  [p_r^2,  x_m \psi^{am}] + x_m \partial_r F_{jlk} [ \psi^{br} \psi_b^j \bar{\psi}_d^l \bar{\psi}^{dk}- \psi_b^r \bar{\psi}^{bj} \delta^{lk}, \psi^{am}] - \frac{1}{2}\delta^{nm} F_{rnm} \psi^{ar}\\
&= i p_r \psi^{ar} - F_{rjl} \psi^{br} \psi_b^j \bar{\psi}^{al}= iQ^a, 
\end{align*}
as required for the corresponding relation  (\ref{odd-even}a). Further on, we have that $[T^{22}, Q^a]=0$ and similarly, $[T^{22}, \bar{Q}_a]=0$, (cf. Lemma \ref{oddevenamanyparticle1}). The remaining relations can be checked similarly. 
\end{proof}

\begin{lemma}
Let $T^{ab}$, $I^{ab}$, $J^{ab}$ be given by \eqref{K}, \eqref{newD},  \eqref{newH},  \eqref{I}, \eqref{J}. Then relations (\ref{thezerorelations}) hold. 
\end{lemma}
The proof of the lemma is the same as the proof of Lemma \ref{thezerothrelationslemma} for the first  representation since $I^{ab}$ and $J^{ab}$ keep the same form, and the proof of commutation relations with $H$ in Lemma~\ref{thezerothrelationslemma} relies only on relations (\ref{odd}) which express $H$ as the anticommutator of the supercharges $Q^a$ and $\bar Q_a$.

\section{Hamiltonians}
\label{Hamsect}

We now proceed to explicit calculations of Hamiltonians appearing in  Theorem~\ref{hamiltonian1manyparticle} and Theorem \ref{hamiltonianmanyparticle2}. We start with a Coxeter root system case. 

\subsection{Coxeter systems}

In this case we take $\mathcal R$ to be a Coxeter root system in $V\cong \mathbb R^N$ \cite{cox}. More exactly, let $\mathcal  R$ be a collection of vectors which spans $V$ and is invariant under orthogonal reflections 
about all the hyperplanes $(\gamma, x)=0$, $\gamma\in \mathcal  R$, where $(\cdot,  \cdot)$ is the standard scalar product in $V$. We also assume that $\mathcal R$ can be decomposed as a disjoint union of its subsets $\mathcal  R_+$ and $- \mathcal R_+$ such that each subsystem $\mathcal  R_+$ and $\mathcal R_-$ contains no collinear vectors. Furthermore, let us assume that squared length $(\gamma, \gamma)=2$ for any $\gamma \in  \mathcal R$, and that $\mathcal R$ is irreducible. Non-equal choices of length of roots in the cases when the Coxeter group has two orbits on $\mathcal R$ are covered by considerations in Subsection \ref{genralV} below.

The corresponding function $F$ has the form
\begin{equation}
\label{prepotCox}
F(x_1, \ldots, x_N) = \frac{\lambda}{2}\sum_{\gamma \in \mathcal R_+} (\gamma, x)^2 \log (\gamma, x),
\end{equation}
where $\lambda \in {\mathbb C}$. 
It is established in \cite{MG}, \cite{Veselov} that $F$ satisfies generalized WDVV equations \eqref{WDVV}. 

Recall the following property. 
\begin{lemma}[Chapter $5$, \cite{bourbaki}]\label{coxeterno}
For any $u, v \in V$ 
\begin{align*}
\sum_{\gamma \in \mathcal R_+} (\gamma,u) (\gamma, v) = h (u, v),
\end{align*}
where $h$ is the Coxeter number of $\mathcal R$. 
\end{lemma}
Lemma \ref{coxeterno} has the following corollary.

\begin{lemma}\label{constraintonF} Let $F$ be given by (\ref{prepotCox}). Then
\begin{align*}
x_i F_{ijk} = \lambda h \delta_{jk}.
\end{align*}
%
%
\end{lemma}

\begin{proof}
Let $\gamma\in \mathcal R$ have coordinates $\gamma = (\gamma_1, \ldots, \gamma_N)$. By   Lemma \ref{coxeterno} we have
\begin{align*}
x_i F_{ijk} = \lambda \sum_{\gamma \in \mathcal R_+} \frac{x_i \gamma_i \gamma_j \gamma_k}{(\gamma, x)} = \lambda \sum_{\gamma \in \mathcal R_+} \gamma_j \gamma_k = \lambda h (e_j, e_k)= \lambda h \delta_{jk}. 
\end{align*} 
\end{proof}

The following identity will be useful below:
\begin{align}\label{thezeroterm}
\sum_{\substack{\beta, \gamma \in \mathcal R_+ \\ \beta \neq \gamma}} \frac{(\beta, \gamma)}{(\beta, x)(\gamma, x)}= 0. 
\end{align}
It follows from the observation that the left-hand side is non-singular at all the hyperplanes $(\beta,x)=0$, $\beta \in \mathcal R_+$.

Let us choose now
\begin{equation}
\label{alpha}
\alpha = -\frac{h \lambda +1}{2}.
\end{equation}
Then $h\lambda = -(2\alpha+1)$, so by Lemma  \ref{constraintonF} function $F$ satisfies the required condition \eqref{alphalambda}. Thus it leads to $D(2,1; \alpha)$ superconformal mechanics with the Hamiltonians given by Theorems  \ref{hamiltonian1manyparticle}, \ref{hamiltonianmanyparticle2}. We now simplify these Hamiltonians for the root system case. 

\begin{theorem}
\label{1strepham}
Let function $F$ be given by \eqref{prepotCox}.  
Then the Hamiltonian $H$ given by (\ref{hamiltonianformulamanyparticle1rep})  is supersymmetric with the superconformal algebra $D(2,1; \alpha)$, where $\alpha$ is given by \eqref{alpha}. The rescaled Hamiltonian
 $H_1=4H$ has the form
\begin{align*}
H_1= -\Delta  +  \sum_{\gamma \in \mathcal R_+} \frac{2 \lambda (\lambda+1)}{(\gamma, x)^2} + \Phi,
\end{align*}
where $\Delta = - p^2$ is the Laplacian in $V$ and the fermionic term
\begin{equation}
\label{Phi}
\Phi= 2 \lambda  \sum_{\gamma \in \mathcal R_+} \frac{\gamma_i \gamma_j \gamma_k \gamma_l}{(\gamma, x)^2}  \psi^{bi} \psi_b^j  \bar{\psi}_d^l \bar{\psi}^{dk} - 4 \lambda \sum_{\gamma \in \mathcal R_+} \frac{\gamma_i \gamma_j}{(\gamma, x)^2} \psi_b^i \bar{\psi}^{bj}.
\end{equation}
\end{theorem}

\begin{proof}
By formula (\ref{hamiltonianformulamanyparticle1rep}) we have  that 
\begin{align*}
H= \frac{p^2}{4} + \Psi + U, 
\end{align*}
where potential 
$$
U= - \frac{1}{8} \partial_i F_{jlk} \delta^{ij} \delta^{lk} + \frac{1}{16} F_{ijk} F_{lmn} \delta^{nm} \delta^{jl} \delta^{ik}
$$
and
$$
\Psi = -\frac12 \partial_iF_{jlk} ( \psi^{bi} \psi_b^j  \bar{\psi}_d^l \bar{\psi}^{dk} - \psi_b^i \bar{\psi}^{bj} \delta^{lk} ).
$$
 Let us firstly simplify $U$. We have
\begin{align*}
F_{jlk} = \lambda  \sum_{\gamma \in  \mathcal R_+ } \frac{\gamma_j \gamma_l \gamma_k}{(\gamma, x)}.
\end{align*}
Then
\begin{align}\label{root1}
\partial_i F_{jlk}  \delta^{ij} \delta^{lk} &= - \lambda \sum_{\gamma \in \mathcal R_+} \frac{\gamma_i \gamma_j \gamma_l \gamma_k}{(\gamma, x)^2}  \delta^{ij} \delta^{lk} = - 4 \lambda \sum_{\gamma \in \mathcal R_+} \frac{1}{(\gamma, x)^2}
\end{align}
and
\begin{align}\label{FF}
F_{ijk} F_{lmn} \delta^{nm} \delta^{jl} \delta^{ik} = 4\lambda^2 \sum_{\beta, \gamma  \in \mathcal R_+} \frac{(\beta, \gamma)}{(\beta, x)(\gamma, x)}=   \sum_{\gamma \in \mathcal R_+ } \frac{8 \lambda^2}{(\gamma, x)^2}
\end{align}
because of identity \eqref{thezeroterm}.
The statement follows from formulas (\ref{root1}), \eqref{FF}. 
\end{proof}

The following theorem can be easily checked directly.

\begin{theorem}
\label{2ndrepham}
 For the function $F$ given by \eqref{prepotCox}  the Hamiltonian $H$ given by (\ref{newH})  is supersymmetric with the superconformal algebra $D(2,1; \alpha)$, where $\alpha$ is given by \eqref{alpha}. The rescaled Hamiltonian
 $H_2=4H$ has the form
\begin{align*}
H_2= -\Delta  +  \sum_{\gamma \in \mathcal R_+} \frac{2 \lambda}{(\gamma, x)}\partial_\gamma + \Phi,
\end{align*}
where $\Phi$ is defined by \eqref{Phi}. 
\end{theorem}

\begin{proposition}
Hamiltonians $H_1, H_2$ from Theorems \ref{1strepham}, \ref{2ndrepham} satisfy gauge relation
\begin{align*}
\delta^{-1} \circ H_2 \circ \delta= H_1,
\end{align*}
where $\delta= \prod_{\beta \in \mathcal R_+} (\beta, x)^\lambda$. 
\end{proposition}
The proof follows immediately by making use of the identity \eqref{thezeroterm}.

\begin{remark}We note that the Hamiltonian $H_2$ is not self-adjoint under hermitian involution defined by 
\begin{align*}
\psi^{aj\dagger}= \bar{\psi}_a^j,  \quad p_j^\dagger=p_j, \quad x_j^\dagger=x_j, \quad i^\dagger=-i, \quad \text{and} \quad (AB)^\dagger= B^\dagger A^\dagger
\end{align*}
for any two operators $A, B$. One could have considered another ansatz for $\bar{Q}_a$ so that to obtain self-adjoint Hamiltonian. Namely,  let $Q^{a}$ be as in (\ref{newQ}) and consider hermitian conjugate $(Q^{a})^\dagger $. Let $Q^a$, $(Q^a)^\dagger $ $(a=1, 2)$ be the ansatz for the supercharges. Then $$(Q^a)^\dagger= p_r \bar{\psi}_a^r +i F_{rjk} \psi_a^k \bar{\psi}_b^r \bar{\psi}^{bj}.$$ Note that since $F_{rjk}\psi_a^k \bar{\psi}_b^r \bar{\psi}^{bj}=F_{rjk}(\bar{\psi}_b^r \bar{\psi}^{bj} \psi_a^k -  \bar{\psi}_a^r \delta^{kj})$ we may express $(Q^a)^\dagger$ in terms of $\bar{Q}_a$ (see \eqref{newQbar}) as follows  $$(Q^a)^\dagger= \bar{Q}_a - i F_{lmn} \bar{\psi}_a^l \delta^{nm}.$$ We then have
\begin{align*}
\{Q^a, (Q^c)^\dagger \} &= \{Q^a, \bar{Q}_c\} - i \{ Q^a, F_{lmn} \bar{\psi}_c^l\}\delta^{nm} \\
&= \{Q^a, \bar{Q}_c\} - \psi^{ar} \bar{\psi}_c^l \partial_r F_{lmn}\delta^{nm} -  \psi_c^r \bar{\psi}^{ak}  F_{rkl} F_{lmn}\delta^{nm} +\frac{i}{2} F_{rmn} p_r \delta^a_c \delta^{nm}, 
\end{align*} 
with $\{Q^a, \bar{Q}_c\}$ defined by \eqref{newQQbar}. Then supersymmetry algebra constraint $\{Q^a, (Q^c)^\dagger \}=-2\delta^a_c H$ leads to restrictions $\alpha=-\frac{1}{2}$, or $\alpha=-\frac{h+2}{4}$. In both cases the bosonic part of the Hamiltonian $H$ can be seen to be zero.


\end{remark}

\subsection{General $\vee$-systems}
\label{genralV}

Let us consider a finite collection of vectors $\mathcal A$ in $V\cong {\mathbb C}^N$ such that the corresponding bilinear form
\begin{align*}
G_{\mathcal A}(u, v)= \sum_{\gamma \in  \mathcal A} (\gamma, u) (\gamma, v), \quad u,v \in V
\end{align*}
is non-degenerate. Let us recall what it means that $\mathcal A$ is a $\vee$-system \cite{Veselov}.  
We can assume by applying a suitable linear transformation to $\mathcal A$ that 
$$
G_{\mathcal A}(u, v) = (u, v)
$$
for any $u, v \in V$. In this case 
$\mathcal A$ is a $\vee$-system if 
for any $\gamma \in \mathcal{A}$ and  for any two-dimensional plane $\pi \subset V$ such that $\gamma \in \pi$ one has 
\begin{align*}
\sum_{\beta \in \mathcal{A}\cap \pi} (\beta, \gamma) \beta = \mu \gamma,
\end{align*}
for  some $\mu= \mu( \gamma, \pi) \in \mathbb{C}$. 

Let $F=F_{\mathcal A}(x_1, \ldots, x_N)$ be the corresponding function
\begin{equation}
\label{prepotVee}
F= \frac{\lambda}{2}\sum_{\gamma \in \mathcal{A}} (\gamma, x)^2 \operatorname{log}(\gamma, x),
\end{equation}
where $\lambda \in {\mathbb C}$. Then $F$ satisfies generalised WDVV equations \eqref{WDVV} (see \cite{Veselov}). Furthermore, the condition 
$$
x_i F_{ijk} = -(2\alpha+1)\delta_{jk}
$$
is satisfied if
$$
\alpha = -\frac12 (\lambda +1).
$$
Therefore this leads to $D(2,1; \alpha)$ superconformal mechanics with the Hamiltonians given by Theorems  \ref{hamiltonian1manyparticle}, \ref{hamiltonianmanyparticle2}, which we present explicitly in the following theorem.
\begin{theorem}
\label{generalthm}
Let function $F$ be given by \eqref{prepotVee}.  
Then the Hamiltonian $H$ given by (\ref{hamiltonianformulamanyparticle1rep})  is supersymmetric with the superconformal algebra $D(2,1; \alpha)$, where $\alpha=-\frac12 (\lambda +1)$. The rescaled Hamiltonian
 $H_1=4H$ has the form
\begin{align*}
H_1= -\Delta + \frac{\lambda}{2}\sum_{\gamma \in \mathcal{A}} \frac{(\gamma, \gamma)^2}{(\gamma, x)^2} + \frac{\lambda^2}{4} \sum_{\gamma, \beta \in \mathcal{A}}  \frac{(\gamma, \gamma)(\beta, \beta) (\gamma, \beta)}{(\gamma, x)(\beta, x)}  + \Phi,
\end{align*}
where $\Delta = - p^2$ is the Laplacian in $V$ and the fermionic term
\begin{equation}
\label{Phi_general}
\Phi=  \sum_{\gamma \in \mathcal A}  \frac{2\lambda\gamma_r \gamma_j \gamma_l \gamma_k}{(\gamma, x)^2}  \psi^{br} \psi_b^j  \bar{\psi}_d^l \bar{\psi}^{dk}  - \sum_{\gamma \in \mathcal{A}} \frac{ 2\lambda\gamma_r \gamma_j (\gamma, \gamma)}{(\gamma, x)^2} \psi_b^r \bar{\psi}^{bj}. 
\end{equation}
Furthermore, the Hamiltonian $H$ given by (\ref{newH})  is also supersymmetric with the superconformal algebra $D(2,1; \alpha)$, where $\alpha = -\frac12 (\lambda +1)$ and the rescaled Hamiltonian
 $H_2=4H$ has the form
$$
H_2= -\Delta +  \lambda \sum_{\gamma \in \mathcal{A}}  \frac{(\gamma, \gamma)}{(\gamma, x)} \partial_{\gamma} + \Phi. 
$$
\end{theorem}
The proof is similar to the one in the Coxeter case. The following proposition can also be checked directly. 
\begin{proposition}
Hamiltonians $H_1, H_2$ from Theorem \ref{generalthm} satisfy gauge relation
\begin{align*}
\delta^{-1} \circ H_2 \circ \delta= H_1,
\end{align*}
where $\delta= \prod_{\beta \in \mathcal A} (\beta, x)^{\frac{\lambda}{2}(\beta, \beta)}$.
\end{proposition}

\section{Trigonometric  version}\label{trigsect}

In this section we consider prepotential functions $F=F(x_1, \ldots, x_N)$ of the form
\begin{align}\label{trigprepotentialF}
F = \sum_{\alpha \in \mathcal A} c_\alpha f((\alpha,x)),
\end{align}
where $\mathcal A$ is a finite set of vectors in $V\cong {\mathbb C}^N$, $c_\alpha \in \mathbb C$ are some multiplicities of these vectors, and function $f$ is given by 
$$
f(z) = \frac16  z^3  -\frac14 \text{Li}_3(e^{-2z})
$$
so that $f'''(z) = \coth z$. 

We are interested in the supercharges of the form 
\begin{align*}
Q^a& = p_r \psi^{ar} + i F_{rjk} \langle \psi^{br} \psi_b^j  \bar{\psi}^{ak} \rangle,
\end{align*}
\begin{align*}
\bar{Q}_c &= p_l  \bar{\psi}_c^l + i F_{lmn} \langle  \bar{\psi}_d^l  \bar{\psi}^{dm} \psi_c^n \rangle,
\end{align*}
$a,c=1,2$, which is analogous to the first representation considered in Section \ref{1st rep section}.

Function $F$ should satisfy conditions 
\begin{align}\label{WDVVmod}
F_{rjk}F_{kmn}= F_{rmk}F_{kjn}, 
\end{align}
for all $r, j, m, n =1,\ldots, N$ 
but we no longer assume conditions \eqref{alphalambda}. Then we have the following statement on supersymmetry algebra.
\begin{theorem}\label{hamiltonian1manyparticleTrig}
Let us assume that $F$ satisfies conditions \eqref{WDVVmod}. Then  for all $a, b =1, 2$ we have 
$$
\{Q^a, Q^b\} = \{\bar{Q}_a, \bar{Q}_b\} = 0 \quad \text{and}\quad   
\{Q^a, \bar{Q}_b  \}=-2 H \delta^a_b,
$$
 where the Hamiltonian $H$ is given by 
\begin{align*}
H= \frac{p^2}{4} - \frac{\partial_iF_{jlk}}{2} ( \psi^{bi} \psi_b^j  \bar{\psi}_d^l \bar{\psi}^{dk} - \psi_b^i \bar{\psi}^{bj} \delta^{lk} + \frac{1}{4} \delta^{ij} \delta^{lk} ) + \frac{1}{16} F_{ijk} F_{lmn} \delta^{nm} \delta^{jl} \delta^{ik}. 
\end{align*}
Furthermore, the rescaled Hamiltonian $H_1=4H$ has the form
\begin{align}\label{H1trig}
H_1= - \Delta + \frac{1}{2} \sum_{\alpha \in \mathcal{A}} \frac{c_\alpha(\alpha, \alpha)^2}{\sinh^2(\alpha, x)} + \frac{1}{4} \sum_{\alpha, \beta \in \mathcal{A}}c_\alpha c_\beta (\alpha, \alpha) (\beta, \beta) (\alpha, \beta) \coth(\alpha,x) \coth(\beta, x) + \Phi, 
\end{align}
where $\Delta=-p^2$ is the Laplacian in $V$ and the fermionic term
\begin{align}\label{phitrig}
\Phi=  \sum_{\alpha \in \mathcal{A}}  \frac{2 c_\alpha \alpha_i \alpha_j }{\sinh^2(\alpha, x)}\big(\alpha_l \alpha_k  \psi^{bi} \psi_b^j \bar{\psi}_d^l \bar{\psi}^{dk} -(\alpha, \alpha) \psi_b^i \bar{\psi}^{bj} \big). 
\end{align}
\end{theorem}
The proof of the first part of the theorem is the same as the proof of Theorem \ref{hamiltonian1manyparticle} together with the proof of the relevant part of Lemma \ref{oddmanyparticle1rep}. The proof of formula \eqref{H1trig} is similar to the proof of Theorem \ref{1strepham}.

Let us now consider supercharges of the form 
\begin{align*}
Q^a& = p_r \psi^{ar} + i F_{rjk}  \psi^{br} \psi_b^j  \bar{\psi}^{ak} ,
\end{align*}
\begin{align*}
\bar{Q}_c &= p_l  \bar{\psi}_c^l + i F_{lmn}   \bar{\psi}_d^l  \bar{\psi}^{dm} \psi_c^n ,
\end{align*}
$a,c=1,2$, which is analogous to the second representation considered in Section \ref{2nd rep section}. Then we have the following statement on supersymmetry algebra.
\begin{theorem}\label{hamiltonian2manyparticleTrig}
Let us assume that $F$ satisfies conditions \eqref{WDVVmod}. Then  for all $a, b =1, 2$ we have 
$$
\{Q^a, Q^b\} = \{\bar{Q}_a, \bar{Q}_b\} = 0 \quad \text{and}\quad   
\{Q^a, \bar{Q}_b  \}=-2 H \delta^a_b,
$$
 where the Hamiltonian $H$ is given by 
\begin{align}\label{hamiltonianformu2amanyparticle1rep}
 H= \frac{p^2}{4} - \frac{\partial_rF_{jlk}}{2} ( \psi^{br} \psi_b^j  \bar{\psi}_d^l \bar{\psi}^{dk} - \psi_b^r \bar{\psi}^{bj} \delta^{lk} ) + \frac{i}{4} \delta^{nm} F_{rmn} p_r. 
\end{align}
Furthermore, the rescaled Hamiltonian $H_2=4H$, has the form 
\begin{align}\label{rescaledH_2}
H_2= - \Delta + \sum_{\alpha \in \mathcal{A}}c_\alpha (\alpha, \alpha) \coth(\alpha,x) \partial_\alpha + \Phi, 
\end{align}
where $\Phi$ is the fermionic term defined by \eqref{phitrig}.
\end{theorem}
The proof of the first part of the theorem is the same as the proof of Theorem \ref{hamiltonianmanyparticle2} together with the proof of the relevant part of Lemma \ref{oddmanyparticle2rep}. Then formula \eqref{rescaledH_2} can be easily derived from the form \eqref{hamiltonianformu2amanyparticle1rep} of $H$. 

Let us now assume that $\mathcal{A}=\mathcal{R}$ is a crystallographic root system, and that the multiplicity function $c(\alpha)=c_\alpha$, $\alpha \in \mathcal{R}$ is invariant under the corresponding Weyl group $W$.  For a general root system $\mathcal R$ the corresponding function $F$ does not satisfy equations \eqref{WDVVmod}. For example, if $\mathcal{R}=A_{N-1}$ then relations \eqref{WDVVmod} do not hold. But for some root systems and collections of multiplicities relations \eqref{WDVVmod} are satisfied. 

In the rest of this section we consider such cases when prepotential $F$ satisfying \eqref{WDVVmod} does exist. The corresponding root systems $\mathcal{R}$ have more than one orbit under the action of the Weyl group $W$. We start by simplifying the corresponding Hamiltonians $H_1$ given by \eqref{H1trig}. 

\begin{proposition}\label{H1rootsystem}Let us assume that prepotential $F$ given by \eqref{trigprepotentialF} for a root system $\mathcal{R}$ with invariant multiplicity function $c$ satisfies \eqref{WDVVmod}. Then Hamiltonian \eqref{H1trig} can be rearranged as 
\begin{align}\label{rescaledH_1}
H_1= - \Delta +  \sum_{\alpha \in \mathcal{R_+}}\frac{\widetilde{c_\alpha} }{\sinh^2(\alpha,x)}+ \widetilde{\Phi},
\end{align}
where 
\begin{align*}
\widetilde{c_\alpha} = 
\begin{dcases}
c_\alpha (\alpha, \alpha)^2 \big( 1 +  c_\alpha (\alpha, \alpha)\big), 
\quad \text{if} \quad 2\alpha \notin \mathcal{R}, \\
c_\alpha (\alpha, \alpha)^2 \big( 1 + (\alpha, \alpha) ( c_\alpha + 8 c_{2\alpha})\big), 
\quad \text{if} \quad 2\alpha \in \mathcal{R},
\end{dcases}
\end{align*}
 $\widetilde{\Phi}= \Phi + const$, with $\Phi$ given by \eqref{phitrig} and $\mathcal{R_+}$ is a positive subsystem in $\mathcal{R}$. 
\end{proposition}
Indeed, it is easy to see that for the crystallographic root system $\mathcal{R}$ the term
\begin{align*}
\sum_{\substack{\beta, \alpha \in \mathcal R \\ \beta \not \sim \alpha}} c_\alpha c_\beta (\alpha, \alpha) (\beta, \beta) (\alpha, \beta) \coth(\alpha,x) \coth(\beta, x)
\end{align*}
 is non-singular at $\tanh(\alpha,x)=0$ for all $\alpha \in \mathcal{R}$, hence it is constant. One can show that the Hamiltonian $H_1$ given by \eqref{H1trig} simplifies to the required form.

We now show that solutions to equations \eqref{WDVVmod} exist for the root systems $\mathcal{R}=BC_N$, $\mathcal{R}=F_4$ and $\mathcal{R}=G_2$, with special collections of invariant multiplicities. 

Let $\mathcal{R_+}$ be a positive subsystem in the root system $\mathcal{R}$. For a pair of vectors $a, b \in V$ we define a $2$-form $\mathcal{B}^{(a,b)}_{\mathcal{R}_+}$ by
\begin{align}\label{B2form}
\mathcal{B}^{(a,b)}_{\mathcal{R}_+}= \sum_{\beta, \gamma \in \mathcal{R}_+} c_\beta c_\gamma (\beta, \gamma) B_{\beta, \gamma}(a, b) \beta \wedge \gamma,
\end{align}
where $B_{\alpha, \beta}(a, b)= \alpha \wedge \beta(a, b)  = (\alpha, a)(\beta, b) - (\alpha, b)(\beta, a)$. 
The form $\mathcal{B}^{(a,b)}_{\mathcal{R}_+}$ has good properties with regard to the action of the corresponding Weyl group $W$. Namely, the following statement takes place. 

\begin{proposition}\label{invariance} The $2$-form (\ref{B2form}) is $W$-invariant, that is 
\begin{align}\label{invarianceBform}
w \mathcal{B}^{(a,b)}_{\mathcal{R}_+} = \mathcal{B}^{(wa,wb)}_{\mathcal{R}_+}= \mathcal{B}^{(wa, wb)}_{w \mathcal{R_+}},
\end{align}
for any $w \in W$. 
\end{proposition}

\begin{proof}
Let us choose a simple root $\alpha \in \mathcal{R_+}$. It is sufficient to prove the statement for $w=s_\alpha$. Let us rewrite $\mathcal{B}^{(a,b)}_{\mathcal{R}_+}$ as 

\begin{align*}
\mathcal{B}^{(a,b)}_{\mathcal{R}_+}= 2 c_{\alpha} \sum_{\beta \in \mathcal{R}_+} c_\beta (\alpha, \beta) B_{\alpha, \beta}(a, b) \alpha \wedge \beta + \sum_{\beta, \gamma \in \mathcal{R}_+ \setminus \{\alpha \}} c_\beta c_\gamma (\beta, \gamma) B_{\beta, \gamma}(a, b) \beta \wedge \gamma.
\end{align*}
It is easy to see that for any $\beta, \gamma \in \mathcal{R}$ 
\begin{align}\label{wB}
 B_{\beta, \gamma}(s_\alpha a, s_\alpha b) = B_{s_\alpha \beta, s_\alpha \gamma}(a, b)
\end{align}
since $(u, s_\alpha v)=(s_\alpha u, v)$ for any $u, v \in V$. Let us now apply $s_\alpha$ to equality (\ref{B2form}). Since $s_\alpha (\mathcal{R}_+ \setminus \{\alpha\})=\mathcal{R}_+ \setminus \{\alpha\}$ we have
\begin{align*}
 s_\alpha \mathcal{B}^{(a,b)}_{\mathcal{R}_+}&= - 2 c_{\alpha} \sum_{\beta \in \mathcal{R}_+} c_\beta (\alpha, \beta) B_{\alpha, \beta}(a, b) \alpha \wedge  \beta + \sum_{\beta, \gamma \in  \mathcal{R}_+ \setminus \{\alpha)} c_\beta c_\gamma (\beta, \gamma) B_{\beta, \gamma}(a, b) s_\alpha \beta \wedge  s_\alpha \gamma \\
 &= 2 c_{\alpha} \sum_{\beta \in \mathcal{R}_+} c_\beta (\alpha, \beta) B_{s_\alpha \alpha, s_\alpha \beta}(a, b) \alpha \wedge  \beta + \sum_{\beta, \gamma \in  \mathcal{R}_+ \setminus \{\alpha)} c_\beta c_\gamma (\beta, \gamma) B_{s_\alpha \beta, s_\alpha \gamma}(a, b)  \beta \wedge  \gamma \nonumber \\
 &= \mathcal{B}^{(s_\alpha a, s_\alpha b)}_{\mathcal{R}_+},
\end{align*}
by the relation \eqref{wB}. This proves the first equality in \eqref{invarianceBform}. In order to prove the second equality \eqref{invarianceBform} let us notice that in fact $$\sum_{\beta \in \mathcal{R_+}} c_\beta (\alpha, \beta) B_{\alpha, \beta}(a, b) \alpha \wedge \beta =0.$$ Hence $s_{\alpha} \mathcal{B}^{(a, b)}_{\mathcal{R_+}}= \mathcal{B}^{(s_\alpha a, s_\alpha b)}_{s_\alpha \mathcal{ R_+}}$. 
\end{proof}

Let us derive some conditions for a function $F$ to satisfy equations of the form (\ref{WDVVmod}). Let $F_i$ be the $N\times N$ matrices of third derivatives of $F$, $(F_i)_{lm} = \frac{\partial^3 F}{\partial x_i \partial x_l \partial x_m}$, and for any vector $a=(a_1, \dots, a_N) \in V$ let us denote $F_a = \sum_{i=1}^N a_i F_i$. 

\begin{theorem}\label{geometricconditions}
Let $a, b \in V$. Then the equations 
\begin{align*}
F_a F_b= F_b F_a
\end{align*}
are satisfied if and only if for any positive system $\mathcal{R}_+$ 
\begin{align}\label{B=0}
\mathcal{B}^{(a,b)}_{\mathcal{R}_+}=0. 
\end{align}
\end{theorem}

\begin{proof}
We have 
\begin{align*}
(F_a)_{lk} = \sum_{\alpha \in \mathcal{R}} c_\alpha (\alpha, a) \alpha_l \alpha_k \coth(\alpha, x), 
\end{align*}
and therefore
\begin{align*}
F_a F_b = \sum_{\alpha, \beta \in \mathcal{R}} c_\alpha c_\beta (\alpha, a) (\alpha, b) (\alpha, \beta) \coth(\alpha,x) \coth(\beta,x) \alpha \otimes \beta.
\end{align*}
Hence the equations $[F_a, F_b]=0$ are equivalent to 
\begin{align*}
\sum_{\alpha, \beta \in \mathcal{R}} c_\alpha c_\beta B_{\alpha, \beta}(a, b) (\alpha, \beta) \coth(\alpha,x) \coth(\beta,x) \alpha \otimes \beta=0,
\end{align*}
which can be easily checked to be equivalent to
\begin{align}\label{sumcotalphacotbeta}
\sum_{\alpha, \beta \in \mathcal{R_+}} c_\alpha c_\beta B_{\alpha, \beta}(a, b) (\alpha, \beta) \coth(\alpha,x) \coth(\beta,x) \alpha \wedge \beta=0.
\end{align}
It is easy to see that the sum in the left-hand side of the equality (\ref{sumcotalphacotbeta}) is non-singular at $\tanh(\alpha, x)=0$ for all $\alpha \in \mathcal{R_+}$, hence this sum is always constant. In an appropriate limit in a cone $\coth(\alpha, x) \rightarrow 1$ for all $\alpha \in \mathcal{R_+}$, and therefore the equality (\ref{sumcotalphacotbeta}) is equivalent to the equality
\begin{align*}
\sum_{\alpha, \beta \in \mathcal{R_+}} c_\alpha c_\beta B_{\alpha, \beta}(a, b) (\alpha, \beta)\alpha \wedge \beta=0,
\end{align*}
as required.
\end{proof}

Let $e_i$, $i=1, \dots, N$ be the standard orthonormal basis in $V$. We may express $\mathcal{B}^{(a,b)}_{\mathcal{R}_+}$ in the basis $e_i \wedge e_j$ of $\Lambda^2 V$,
\begin{align}\label{B2formeiej}
\mathcal{B}^{(a,b)}_{\mathcal{R}_+}= \sum_{1 \leq i < j \leq N} g_{ij} e_i \wedge e_j,
\end{align}
for some scalars $g_{ij}=g_{ij}(a,b)$.  Then linear independence of the basis vectors and condition (\ref{B=0}) give rise to ${N \choose 2}$ equations $g_{ij}(a,b)=0$. If $A_{N-1} \subset \mathcal{R}$ then by Proposition \ref{invariance} we should have that $g_{ij}(a,b)=\pm g_{\sigma(i) \sigma(j)}(\sigma(a), \sigma(b))$ for any transposition $\sigma \in S_N$ which acts on vectors $a,b$ by the corresponding permutation of coordinates. This shows that the condition \eqref{B=0} reduces to a single equation $g_{ij}=0$ for any fixed $i, j$ and general $a, b \in V$.  For convenience we will write below $B_{e_i, e_j}(a, b)$ as $B_{ij}(a,b)$.

\begin{theorem}\label{BCN}
Let $\mathcal{R}=BC_N$. Let the positive half of the root system $BC_N$ be 
\begin{align*}
\eta e_i, \, 2 \eta e_i, \, 1 \leq i \leq N; \quad \eta( e_i \pm e_j), \, 1\leq i < j \leq N,
\end{align*}
where $\eta \in \mathbb{C}^{\times}$ is a parameter.  Let $r$ be the multiplicity of vectors $\eta e_i$, and let $s$ be the multiplicity of vectors $2 \eta e_i$. Let $q$ be the multiplicity of vectors $\eta (e_i \pm e_j)$. Then the function
\begin{align}\label{FBCN}
F=  \sum_{i=1}^N (r f(\eta x_i) + s f(2 \eta x_i) ) +  q \sum_{i <j}^N f (\eta(x_i \pm  x_j))
\end{align}
satisfies conditions \eqref{WDVVmod} if and only if $r= -8s - 2 (N-2)q$.  The corresponding supersymmetric Hamiltonians given by \eqref{rescaledH_2}, \eqref{rescaledH_1} take the form 
\begin{align}\label{H1BCN}
H_1& = - \Delta + \eta^4 \sum_{i=1}^N \Big(  \frac{-(8s +2(N-2)q)(1-2 (N-2)q\eta^2)}{\sinh^2 \eta x_i}  + \frac{16s (1 +4s \eta^2 ) }{\sinh^2 2 \eta x_i}\Big)\\ &+ \eta^4\sum_{i < j}^N \frac{4q  (1+2 q\eta^2 )  }{\sinh^2(\eta(x_i \pm x_j))}+ \widetilde{\Phi} \nonumber, 
\end{align}
and
\begin{align}\label{H2BCN}
H_2&=  - \Delta + 2 \eta^3 \sum_{i=1}^N \big(-(8s +2(N-2)q) \coth \eta x_i + 8s \coth 2 \eta x_i \big)\partial_{i} \\&+ 4q \eta^3 \sum_{i < j  }^N  \coth(\eta (x_i \pm x_j))( \partial_{i} \pm \partial_{j}) + \Phi \nonumber, 
\end{align}
with $\Phi$ given by 
\begin{align*}
\Phi&= 4 \eta^4 \sum_{i=1}^N \Big(  \frac{-(8s+2(N-2)q)}{\sinh^2 \eta x_i} + \frac{16s}{\sinh^2 2 \eta x_i} \Big) \Big( \psi^{bi} \psi_b^i \bar{\psi}_d^i  \bar{\psi}^{di} - \psi_b^i  \bar{\psi}^{bi} \Big) \\
&+4 \eta^4 \sum_{\epsilon \in \{1, -1\}}  \sum_{ m < t }^N \sum_{i,j,l,k} \frac{q d_{mti} d_{mtj}}{\sinh^2(\eta(x_m + \epsilon x_t))} \Big( d_{mtl} d_{mtk} \psi^{bi} \psi_b^j \bar{\psi}_d^l  \bar{\psi}^{dk} - 2\psi_b^i  \bar{\psi}^{bj} \Big),
\end{align*}
where $d_{mtk}=d_{mtk}(\epsilon)= \delta_{mk} + \epsilon \delta_{tk}$, and $\widetilde{\Phi}= \Phi + const$.
\end{theorem}

\begin{proof}
Let us use Theorem \ref{geometricconditions} in order to deal with conditions (\ref{WDVVmod}). Let us consider the coefficient $g_{12}(a,b)$ at $e_1 \wedge e_2$ by collecting respective terms in the corresponding form $\mathcal{B}^{(a,b)}_{\mathcal{R_+}}$ given by \eqref{B2form}, \eqref{B2formeiej}. The non-trivial contribution to $g_{12}$ comes only from the following pairs of vectors $\{\beta, \gamma\}$ in the expansion \eqref{B2form}:
\begin{align*}
(1)\, \,  \{ \eta e_1,\eta( e_1 \pm e_2)\},\quad (2)\, \, \{2 \eta  e_1, \eta( e_1 \pm e_2)\} ,\quad (3)\, \,  \{\eta(e_1 \pm e_2), \eta(e_1 \pm e_j)\}, \, \,  3 \leq j \leq N,
\end{align*}
since contributions from pairs $\{\eta(e_1 \pm e_2), \eta(e_2 \pm e_j)\}$ and $\{ \eta(e_1 \pm e_j), \eta(e_2 \pm e_j)\}$ is zero each. Pairs $(1)$ contribute  $4  rq \eta^6 B_{12}(a,b)$, pairs $(2)$ contribute $32 sq \eta^6 B_{12}(a,b) $ and pairs $(3)$ contribute $8 q^2 (N-2) \eta^6 B_{12}(a,b)$. Therefore 
\begin{align*}
g_{12}(a, b)=  4q (r + 8s +2(N-2)q)\eta^6  B_{12}(a, b).
\end{align*}
By Proposition \ref{invariance}, $g_{ij}=0$ for all $1 \leq i < j \leq N$ if and only if $ r=-8s -2(N-2)q$. The form of the Hamiltonians $H_2$, $H_1$ follows from Theorem \ref{hamiltonian2manyparticleTrig} and Proposition \ref{H1rootsystem} respectively. Then the statement follows. 
\end{proof}

\begin{remark}
We note that for the multiplicity $s=0$ Theorem \ref{BCN} is contained in \cite{HM}. Indeed, Theorem $2.3$ in \cite{HM} states that the function $F$ given by formula \eqref{FBCN} with root system $\mathcal{R}=B_N$ satisfies WDVV equations. It also follows from the proof of Theorem 2.3 in \cite{HM} that the corresponding metric is proportional to the standard metric $\delta_{ij}$. Therefore WDVV equations are equivalent to equations \eqref{WDVVmod}.
\end{remark}

\begin{theorem}
Let $\mathcal{R}=F_4$. Let the positive half of the root system $F_4$ be
\begin{align*}
\eta e_i,  \, 1 \leq i \leq 4; \quad \eta (e_i \pm e_j), \, 1\leq i < j \leq 4; \quad  \frac{\eta}{2} (e_1 \pm  e_2 \pm e_3 \pm e_4 ),
\end{align*}
where $\eta \in \mathbb{C}^\times$ is a parameter. Let $r$ be the multiplicity of short roots $\eta e_i$,  $\frac{\eta}{2} (e_1 \pm  e_2 \pm e_3 \pm e_4 )$ and let $q$ be the multiplicity of long roots $\eta(e_i \pm e_j)$. The function 
\begin{align*}
F=  r \sum_{i=1}^4 f(\eta x_i) +r\sum_{\epsilon_i \in \{1,-1\}} f(\frac{\eta}{2} (\epsilon_1 x_1 + \epsilon_2  x_2 + \epsilon_3  x_3 + x_4))+  q \sum_{i <j}^4 f (\eta(x_i \pm  x_j))
\end{align*}
 satisfies conditions \eqref{WDVVmod} if and only if $r=-2q$ or $r=-4q$. The corresponding supersymmetric Hamiltonians \eqref{rescaledH_1}, \eqref{rescaledH_2} take the form 
\begin{align*}
H_1 &= - \Delta +  r(1+ r\eta^2)\eta^4 \Big( \sum_{i=1}^4 \frac{1}{\sinh^2 \eta x_i} +\sum_{\epsilon_i \in \{1, -1\}}\frac{1}{\sinh^2(\frac{\eta}{2}(\epsilon_1 x_1 +\epsilon_2 x_2 + \epsilon_3 x_3 + x_4))} \Big) \\ &+ \eta^4 \sum_{ i < j }^4 \frac{4q  (1+2 q\eta^2)}{\sinh^2(\eta(x_i \pm x_j))} +\widetilde{\Phi},
\end{align*}
and
\begin{align*}
H_2 &= - \Delta + r \eta^3   \sum_{\epsilon_i \in \{1, -1\}} \coth(\frac{\eta}{2}(\epsilon_1 x_1 + \epsilon_2  x_2 + \epsilon_3 x_3 +  x_4))  (\epsilon_1 \partial_{1} + \epsilon_2  \partial_{2}  + \epsilon_3  \partial_{3}  + \partial_{4} ) \\ &+2r \eta^3  \sum_{i=1}^4 \coth \eta x_i \partial_{ i} + 4q \eta^3 \sum_{ i < j}^4 \coth(\eta( x_i \pm x_j)) (\partial_{i} \pm \partial_{j}) +   \Phi
\end{align*}
with $\Phi$ given by
\begin{align*}
\Phi&= 4 \eta^4 \sum_{i=1}^4  \frac{r}{\sinh^2 \eta x_i}  \Big( \psi^{bi} \psi_b^i \bar{\psi}_d^i  \bar{\psi}^{di} - \psi_b^i  \bar{\psi}^{bi} \Big) \\
&+4\eta^4 \sum_{\epsilon \in \{1, -1 \}}  \sum_{ m < t }^4\sum_{i,j,l,k} \frac{q d_{mti} d_{mtj}}{\sinh^2 \eta((x_m + \epsilon x_t))} \Big( d_{mtl} d_{mtk} \psi^{bi} \psi_b^j \bar{\psi}_d^l  \bar{\psi}^{dk} - 2\psi_b^i  \bar{\psi}^{bj} \Big) \\
&+ 4\eta^4 \sum_{\epsilon_i \in \{1, -1\}} \sum_{i,j,l,k} \frac{r d_i d_j}{\sinh^2(\frac{\eta}{2} (\epsilon_1 x_1 + \epsilon_2 x_2 + \epsilon_3 x_3 + x_4))}\Big(d_l d_k \psi^{bi} \psi_b^i \bar{\psi}_d^i  \bar{\psi}^{di} - \psi_b^i  \bar{\psi}^{bi} \Big) ,
\end{align*}
where $r=-2q$ or $r=-4q$, $d_i=d_i(\epsilon_1, \epsilon_2, \epsilon_3)= \epsilon_1 \delta_{1i} + \epsilon_2  \delta_{2i}+ \epsilon_3 \delta_{3i} + \delta_{4i}$ and $\widetilde{\Phi}= \Phi + const$. 
\end{theorem}

\begin{proof}
Since $B_4 \subset F_4$ we have the contribution to the coefficient $g_{12}$ of the form \eqref{B2form}, \eqref{B2formeiej} from the pairs of vectors $\{\beta, \gamma\} \in B_4$ which is equal to $4q(4q+r) \eta^6 B_{12}(a, b)$.  The remaining contribution to the coefficient $g_{12}$ comes from the following pairs of vectors $\{\beta, \gamma\}$ in the expansion \eqref{B2form}:
\begin{align*}
(1)\,\, \{\eta e_1, \frac{\eta}{2} (e_1 \pm  e_2 \pm e_3 \pm e_4 )\}, \quad (2)\,\, \{ \eta( e_1 \pm e_3), \frac{\eta}{2} (e_1 \pm  e_2 \pm e_3 \pm e_4 )\},
\end{align*}
\begin{align*}
\quad (3)\,\, \{ \eta(e_1 \pm e_4), \frac{\eta}{2} (e_1 \pm  e_2 \pm e_3 \pm e_4 )\}.
\end{align*}
Indeed, let us demonstrate why pairs of vectors of the form  
\begin{align} \label{F4vect.}
\frac{\eta}{2} (e_1 \pm  e_2 \pm e_3 \pm e_4 )
\end{align}
contribute trivially to the coefficient $g_{12}$ of the form \eqref{B2form}, \eqref{B2formeiej}. Let $\beta= \frac{\eta}{2} (e_1 +\lambda e_2 + \mu e_3 + \nu e_4)$ and $\widetilde{\beta} = \frac{\eta}{2} (e_1 +\lambda e_2 - \mu e_3 - \nu e_4)$, where $\lambda, \mu, \nu = \pm 1$. Non-trivial contribution with this $\beta$ to $g_{12}$ can only come from the two pairs $\{\beta, \pm \gamma\}$, where $\gamma_{\pm} = \frac{\eta}{2}(e_1 - \lambda e_2 \pm (\mu e_3 + \nu e_4))$. The same holds for $\widetilde{\beta}$. The contribution from the two pairs $\{\beta, \gamma_{\pm}\}$ is $-\frac{\lambda r^2}{4} \eta^6 B_{e_1+\lambda e_2, \mu e_3 + \nu e_4}$ while the contribution from the two pairs $\{\widetilde{\beta}, \gamma_{\pm}\}$ is $\frac{\lambda r^2}{4}\eta^6 B_{e_1+\lambda e_2, \mu e_3 + \nu e_4}$. Hence altogether contributions to $g_{12}$ from pairs of vectors of the form (\ref{F4vect.}) cancel. Similarly, one can check that contributions from pairs $\{ \eta e_2, \frac{\eta}{2}(e_1 \pm e_2 \pm e_3 \pm e_4)\}$ and $\{ \eta(e_1 \pm e_2 ) , \frac{\eta}{2}(e_1 \pm e_2 \pm e_3 \pm e_4)\}$ is zero. 

Then pairs $(1)$ contribute $2 r^2 \eta^6 B_{12}(a,b)$ and pairs $(2), (3)$ contribute $4 rq \eta^6 B_{12}(a,b)$ each. Therefore in total 
\begin{align*}
g_{12}(a, b)=2(8q^2 + 6rq + r^2)\eta^6  B_{12}(a, b). 
\end{align*}
By Proposition \ref{invariance}, $g_{ij}=0$ for all   $1 \leq i < j \leq 4$ if and only if  $r=-2q$ or $r=-4q$. The form of the Hamiltonians $H_2$, $H_1$ follows from Theorem \ref{hamiltonian2manyparticleTrig} and Proposition \ref{H1rootsystem}. Then the statement follows.
\end{proof}

\begin{theorem}
Let $\mathcal{R}=G_2$. Let the positive half of the root system $G_2$ considered in three dimensional space be 
\begin{align*}
\quad \alpha_1= \eta( e_1 -e_2), \quad  \alpha_2=\eta( e_1-e_3), \quad  \alpha_3= \eta(e_2-e_3),  
\end{align*}
\begin{align*}
\alpha_4=\eta(2e_1 -e_2 -e_3), \quad \alpha_5=\eta(e_1 +e_2 -2e_3), \quad \alpha_6=\eta(e_1 -2e_2 +e_3),
\end{align*}
where $\eta \in \mathbb{C}^\times$ is a parameter. Let $s$ be the multiplicity of the short roots $\alpha_i$, $i=1, 2, 3$ and let $r$ be the multiplicity of the long roots $\alpha_j$, $j=4,5,6$. Then the function
\begin{align*}
F=  s \sum_{i < j}^3 f(\eta(x_i - x_j)) +  \frac{r}{2} \sum_{\sigma \in S_3} f(\eta(2x_{\sigma{(1)}} -x_{\sigma{(2)}} -x_{\sigma{(3)}}))
\end{align*}
satisfies conditions \eqref{WDVVmod} if and only if $ s=-3r$ or $s=-9r$. The corresponding supersymmetric Hamiltonians \eqref{rescaledH_1}, \eqref{rescaledH_2} take the form 
\begin{align*}
H_1&= - \Delta +\eta^4 \sum_{ i < j }^3 \frac{4s(1+2s\eta^2 )}{\sinh^2(\eta(x_i - x_j))} +  \eta^4 \sum_{\sigma \in S_3} \frac{18r (1+6r \eta^2) }{\sinh^2(\eta(2x_{\sigma(1)} -x_{\sigma{(2)}} -x_{\sigma(3)}))}+ \widetilde{\Phi},
\end{align*}
and 
\begin{align*}
H_2 &= - \Delta + 4s \eta^3  \sum_{ i < j }^3  \coth(\eta(x_i - x_j)) (\partial_{i} - \partial_{j}) \\&+ 6r\eta^3 \sum_{\sigma \in S_3} \coth(\eta(2x_{\sigma(1)} -x_{\sigma(2)} -x_{\sigma(3)})) (2 \partial_{\sigma(1)} - \partial_{\sigma(2)} - \partial_{\sigma(3)}) +  \Phi, 
\end{align*}
with $\Phi$ given by 
\begin{align*}
\Phi&= 4\eta^4 \sum_{ m < t }^3\sum_{i,j,l,k} \frac{s d^{-}_{mti} d^{-}_{mtj}}{\sinh^2(\eta(x_m - x_t))} \Big( d^{-}_{mtl} d^{-}_{mtk} \psi^{bi} \psi_b^j \bar{\psi}_d^l  \bar{\psi}^{dk} - 2\psi_b^i  \bar{\psi}^{bj} \Big) \\
&+2\eta^4\sum_{\sigma \in S_3} \sum_{i,j,l,k}  \frac{r d^\sigma_i d^\sigma_j}{\sinh^2(\eta(2x_{\sigma(1)} -x_{\sigma{(2)}} -x_{\sigma{(3)}})) }\Big(d^\sigma_l d^\sigma_k \psi^{bi} \psi_b^j \bar{\psi}_d^l  \bar{\psi}^{dk} - 6\psi_b^i  \bar{\psi}^{bj} \Big)
\end{align*} 
where $s=-3r$ or $s=-9r$,  $d^{-}_{mti}= \delta_{mi} - \delta_{ti}$,  $d^\sigma_i= 2 \delta_{\sigma(1)i} - \delta_{\sigma(2)i} - \delta_{\sigma(3)i}$, and $\widetilde{\Phi}= \Phi + const$. 
\end{theorem}
\begin{proof}
The coefficient at $e_1 \wedge e_2$ in the form $\mathcal{B}^{(a, b)}_{\mathcal{R_+}}$ given by \eqref{B2form}, \eqref{B2formeiej} is 
\begin{align*}
g_{12}(a,b)&= \sum_{ i < j }^6 2 c_{\alpha_i} c_{\alpha_j} (\alpha_i, \alpha_j) B_{\alpha_i, \alpha_j}(a, b) (\alpha_i \wedge \alpha_j, e_1 \wedge e_2)= \sum_{i=1}^5 A_i, 
\end{align*}
where $(\alpha_i \wedge \alpha_j, e_1 \wedge e_2) = \operatorname{det}(c_1, c_2)$ where $c_k$ are the column vectors $c_k=((\alpha_i, e_k), (\alpha_j, e_k))^\intercal$, $k=1, 2$, and
\begin{align*}
A_i = \sum_{j=i+1}^6  2 c_{\alpha_i} c_{\alpha_j} (\alpha_i, \alpha_j) B_{\alpha_i, \alpha_j}(a, b) (\alpha_i \wedge \alpha_j, e_1 \wedge e_2). 
\end{align*}
We have 
\begin{align*}
A_1 &=  6 s r  \eta^6 B_{\alpha_1, \alpha_5}(a, b),\\
A_2 &=  2s \eta^6\big(  s B_{\alpha_2, \alpha_3}(a,b)  - 3 r B_{\alpha_2, \alpha_6}(a, b) \big), \\
A_3 &= 0, \\
A_4 &= 18  r^2\eta^6 B_{\alpha_4, 3\alpha_3}(a,b) ,\\
A_5 &= 18 r^2\eta^6 B_{\alpha_5, \alpha_6}(a,b).
\end{align*}
Simplifying we obtain
\begin{align*}
g_{12}(a,b)= 2\eta^6(27r^2 +12rs + s^2) ( B_{12}(a,b) -B_{13}(a,b) + B_{23}(a,b)). 
\end{align*}
By Proposition \ref{invariance}, $g_{ij}=0$ for all $1 \leq i <j \leq 3$ if and only if  $s=-3r$ or  $s= -9r$. 
The form of the Hamiltonians $H_1$, $H_2$ follows from Theorem \ref{hamiltonian2manyparticleTrig} and Proposition \ref{H1rootsystem} respectively. Then the statement follows. 
\end{proof}

\begin{remark}The bosonic part of the supersymmetric Hamiltonians \eqref{rescaledH_2}, \eqref{rescaledH_1} becomes Calogero--Moser Hamiltonian in the rational limit. For example let us consider the case of the root system $BC_N$ and let us introduce rescaled multiplicities $\widehat{s}=\eta^2s$, $ \widehat{q}=\eta^2 q $ and $ \widehat{r}=\eta^2 r $ in Theorem \ref{BCN}. Then in the limit $\eta \rightarrow 0$ bosonic parts of Hamiltonians $H_1$ and $H_2$ given by \eqref{H1BCN}, \eqref{H2BCN} become the rational $B_N$ Hamiltonians $H_1^{b, r}, H_2^{b, r}$ with two independent coupling parameters, namely, 
\begin{align*}
H_1^{b,r}=- \Delta+ \sum_{ i <j}^N \frac{4 \widehat{q} (2\widehat{q}+1)}{(x_i \pm x_j)^2}  + \sum_{i=1}^N \frac{l (l-1)}{x_i^2},
\end{align*}
and
\begin{align*}
H_2^{b,r} = -\Delta + \sum_{ i <j}^N \frac{4 \widehat{q}}{x_i \pm x_j} (\partial_{i} \pm \partial_{j})  -\sum_{i=1}^N \frac{2l}{x_i } \partial_{i},
\end{align*}
where $l= 2( (N-2)\widehat{q} + 2\widehat{s} )$. Thus supersymmetric Hamiltonians \eqref{H1BCN}, \eqref{H2BCN} can be viewed as $\eta$-deformation of the rational superconformal Hamiltonians considered in Theorems \ref{1strepham}, \ref{2ndrepham} for the root system ${\mathcal R} = B_N$.
\end{remark}

\section{Concluding remarks}

 Since work \cite{Wyllard} there were extensive attempts to define superconformal ${\mathcal N}=4$ Calogero--Moser type systems for sufficiently general coupling parameters and suitable superconformal algebras. Some low rank cases were treated in \cite{GLP07}, \cite{GLP09}. A number of works were devoted to the superconformal extensions of Calogero--Moser systems where extra spin type variables had to be present (see \cite{FIL12review} for a discussion and the review). In the current work we presented superconformal extensions of the ordinary Calogero--Moser system with scalar potential as well as its generalisations for an arbitrary $\vee$-system, which includes Olshanetsky--Perelomov generalisations of Calogero--Moser systems with arbitrary invariant coupling parameters. The superconformal algebra is $D(2,1;\alpha)$ where parameter $\alpha$ is related to the coupling parameter(s). It is crucial for our considerations that we deal with quantum rather than classical Calogero--Moser type systems.

We also presented supersymmetric non-conformal deformations of the Calogero--Moser type systems related with the root system $B_N$ (which may be thought of as the Calogero--Moser system with boundary terms) as well as with some other exceptional root systems. It would be very interesting to see if there are any relations of considered systems with black holes (cf. \cite{TS} for the conjectural relation with supersymmetric Calogero--Moser systems and e.g. \cite{MSY}, \cite{OY} and references therein for non--conformal deformations of $AdS_2$ black hole geometry).

All our considerations are also extended to non-self-adjoint gauge of the Calogero--Moser type Hamiltonians. There has been considerable interest in such non-self-adjoint but $\mathcal P \mathcal T$ symmetric bosonic Hamiltonians (see e.g. \cite{FZ} and references therein). It would be interesting to see whether these Hamiltonians play a role in the context of supersymmetry. 

It may also be interesting to clarify integrability of considered supersymmetric Hamiltonians.


\end{document}